\documentclass[11pt]{article}
\usepackage[margin=1in]{geometry}
\usepackage{amsmath,amssymb,amsthm,bm,mathtools}
\usepackage{graphicx,booktabs,float}
\usepackage{algorithm}
\usepackage{algpseudocode}
\usepackage{natbib}
\usepackage[hidelinks]{hyperref}
\usepackage{caption}
\usepackage{authblk}
\graphicspath{{../simulations/outputs/}{../report/}}
\newtheorem{proposition}{Proposition}
\newtheorem{assumption}{Assumption}
\newtheorem{theorem}{Theorem}
\newtheorem{corollary}{Corollary}
\newtheorem{lemma}{Lemma}
\DeclareMathOperator{\N}{N}
\DeclareMathOperator{\IW}{IW}
\newcommand{\R}{\mathbb{R}}\newcommand{\E}{\mathbb{E}}
\newcommand{\Sig}{\Sigma}\newcommand{\Had}{\odot}
\newcommand{\bx}{\bm{x}}\newcommand{\by}{\bm{y}}\newcommand{\bZ}{\bm{Z}}
\newcommand{\btau}{\bm{\tau}}\newcommand{\bmu}{\bm{\mu}}\newcommand{\beps}{\bm{\varepsilon}}
\newcommand{\ind}{\mathbf{1}}
\providecommand{\pkg}[1]{\textsf{#1}}

\usepackage{tikz}
\usetikzlibrary{arrows.meta,positioning}

\title{\bf Factorial Multivariate Bayesian Causal Forests:\\
heterogeneous main and interaction effects of multiple treatments\\ on correlated outcomes}
\author[1]{Danilo A. Sarti}
\affil[1]{\small Royal College of Surgeons in Ireland (RCSI University of Medicine and Health Sciences), Dublin, Ireland}
\date{\today}

\begin{document}
\maketitle

\begin{abstract}\noindent
Many studies expose units to several binary treatments at once and record several
correlated outcomes, yet analysts usually estimate one treatment's average effect on
one outcome at a time --- discarding how treatments interact and how their effects vary
across units. We introduce the factorial multivariate Bayesian causal forest, a
Bayesian nonparametric model that decomposes the factorial response surface over the
treatment lattice into a prognostic sum-of-trees plus one sum-of-trees per main or
interaction effect, each with correlated multivariate leaf parameters, estimated
jointly with coherent uncertainty. A single indicator-weighted kernel samples every
component --- the prognostic term is the special case whose indicator is unity --- so the
two-treatment model, the single-treatment multivariate causal forest, and a general
order-$r$ truncation are one and the same sampler. We give the analysis-of-variance
(ANOVA)/M\"obius identification of each estimand, an efficient \texttt{Rcpp} engine, and
an interpretability layer using value-suppressing uncertainty maps. In simulations the
average-effect estimators are unbiased with nominal coverage, consistent, and robust
under misspecification, failing only under unmeasured confounding, which we flag. We
illustrate the method on clinical, agricultural and economic data and on a deeper
application to the NHANES survey. Software is provided as an \textsf{R} package.

\smallskip\noindent Keywords: Bayesian causal forest; factorial
experiments; treatment interaction; heterogeneous treatment effects; multivariate
outcomes; BART.
\end{abstract}

\section{Introduction}
Consider an agronomist testing, across many environments, whether a biological
inoculant buffers the yield loss a crop suffers under drought --- and for which
genotypes, measuring yield, grain weight and height together, the classic
genotype-by-environment problem \citep{finlay1963,gauch2006}. Or a trialist asking
whether two adjuvant therapies interact on several clinical endpoints, as in the
factorial-trial tradition \citep{dasgupta2015factorial}. Or an economist asking
whether a neighbourhood premium differs for air-conditioned homes, the hedonic
question of \citet{anglin1996hedonic}.
Each is a question about multiple crossed treatments, their interaction,
and heterogeneity across units, on several correlated outcomes. None
is answered by estimating one average effect on one outcome.

These are not niceties of presentation but the substance of the decision at stake.
When treatments are studied one at a time, a real interaction is silently absorbed
into the main effects, so an intervention that looks beneficial on average can be
useless or harmful in the subgroup where a second treatment is already present ---
precisely the regime a practitioner needs to know about. When correlated outcomes
are modelled separately, the analyst both loses efficiency and forfeits the joint
pattern that clinicians, breeders and policymakers actually act on, while the
one-outcome-at-a-time multiplicity quietly inflates false discoveries. The
recommendation that matters is rarely an average: it is which inoculant helps which
genotype under stress, whether adding a second drug helps or interferes across a
patient's whole clinical profile, and for whom an economic premium holds. Answering
questions of that shape from observational data, with uncertainty that is honest
about where the design runs thin, is the practical problem we take up.

The tools that come closest each stop short of this combination.
Bayesian causal forests (BCF) \citep{hahn2020bcf} estimate a single binary treatment's
heterogeneous effect with excellent finite-sample behaviour, and were recently
extended to multivariate outcomes as the multivariate BCF (MVBCF) \citep{mcjames2025mvbcf}, but for one
treatment only. The factorial-experiments literature decomposes several treatments'
effects into main and interaction terms \citep{egami2019causal,dasgupta2015factorial}
and even models heterogeneity in high-dimensional (conjoint) treatments
\citep{goplerud2022factorhet}, but through linear or regularized models, for
a single outcome, and often for average rather than fully
nonparametric heterogeneous effects. Generalized random forests provide a
multi-arm, multi-outcome estimator \citep{athey2019grf}, but as frequentist
arm-versus-baseline contrasts, without an explicit ANOVA interaction decomposition,
heredity shrinkage, or a coherent joint posterior over the interaction surface.

We fill exactly this gap with a model that delivers, in
one coherent object: (i) heterogeneous main and interaction effects,
nonparametrically; (ii) over multivariate correlated outcomes, borrowing
strength through a shared residual covariance; (iii) via the ANOVA/M\"obius
decomposition of the treatment lattice with heredity-style shrinkage on higher
orders; (iv) with coherent Bayesian uncertainty over the whole surface; and
(v) a single unified sampler in which the prognostic forest is the degenerate,
indicator-$\equiv\!1$ case of an effect forest (Proposition~\ref{prop:reduce}).
The two-treatment model, the single-treatment multivariate causal forest, and the
general $T$-treatment order-$r$ truncation are the same engine at different
component sets. We supply identification results, an \texttt{Rcpp} implementation,
an interpretability layer, an extensive simulation and robustness study, and
cross-domain real-data illustrations. In one sentence, we extend the
multivariate Bayesian Causal Forest of \citet{mcjames2025mvbcf} from a single
treatment to the full factorial of several crossed treatments, adding the
interaction lattice and its heredity shrinkage, and show (Proposition~\ref{prop:reduce})
that the whole family --- prognostic, main effects and interactions --- is sampled by
one indicator-weighted kernel. We claim the combination and its unification,
not any single ingredient. The model as developed here assumes a Gaussian likelihood and so targets continuous outcomes; Appendix~\ref{app:nongaussian} shows how the same indicator-weighted kernel extends to binary, count and censored-survival endpoints on a latent scale.

Table~\ref{tab:positioning} places the method against its
closest relatives; none delivers all of these at once --- heterogeneous main
and interaction effects, nonparametrically, over multivariate correlated
outcomes, with coherent Bayesian uncertainty, from a single unified sampler.

\begin{table}[H]\centering\small
\caption{Positioning against the closest methods. The proposed model occupies the
intersection: multivariate outcomes $\times$ factorial interaction $\times$ Bayesian
nonparametric, unified in one kernel.}
\label{tab:positioning}
\begin{tabular}{@{}p{3.1cm}p{5.0cm}p{5.6cm}@{}}
\toprule
Method & What it provides & What we add \\
\midrule
BCF \citep{hahn2020bcf} & single treatment, single outcome, heterogeneous; propensity/prognostic split & the factorial lattice (main + interaction) and multivariate outcomes \\
MVBCF \citep{mcjames2025mvbcf} (our base) & multivariate correlated outcomes, single treatment & $\ge 2$ crossed treatments and their interaction lattice \\
Egami--Imai \citep{egami2019causal} & factorial interaction decomposition; linear, design-based, average, single outcome & nonparametric heterogeneity, multivariate outcomes, Bayesian uncertainty \\
FactorHet \citep{goplerud2022factorhet} & heterogeneous factorial effects; regularized mixture, single outcome & tree-ensemble heterogeneity, multivariate outcomes, coherent posterior \\
grf multi-arm \citep{athey2019grf} & frequentist arm-vs-baseline contrasts, multi-outcome & explicit ANOVA interaction with heredity shrinkage; coherent Bayesian posterior over the interaction surface \\
\bottomrule
\end{tabular}
\end{table}

\section{Related work}
Our method sits at the confluence of several literatures; we summarise each and
state what we add. Broad accounts of causal machine learning appear in
\citet{athey2017state,yao2021survey,knaus2021,li2023bayesreview,linero2023wires},
and for Bayesian additive regression trees (BART) specifically in \citet{hill2020review}.

We work in the
Neyman--Rubin potential-outcomes framework \citep{rubin1974,imbens2015}, adjusting
for confounding through the propensity score \citep{rosenbaum1983} and its
generalisation to non-binary treatments \citep{imbens2000}; the prognostic/effect
residualisation is in the spirit of partially linear models \citep{robinson1988}.

BART \citep{chipman2010bart},
building on Bayesian classification-and-regression trees (CART) \citep{chipman1998,chipman2002}, is our building block: we
use its Gaussian-conjugate leaf updates and depth prior and inherit refinements for
smoothness and sparsity \citep{linero2018soft,linero2018sparse}, general/likelihood
extensions \citep{tan2019general}, efficient software \citep{sparapani2021}, and its
posterior-concentration theory \citep{rockova2020}.

BART was brought to causal
inference by \citet{hill2011} and to heterogeneous effects in experiments by
\citet{green2012}; the Bayesian Causal Forest \citep{hahn2020bcf} added the
prognostic/effect split and propensity control, with shrinkage variants
\citep{caron2022shrinkage} and strong data-competition performance
\citep{dorie2019}. \citet{mcjames2025mvbcf} extended BCF to multivariate
correlated outcomes for a single treatment --- the model we directly generalise.
Broader Bayesian-causal perspectives appear in
\citet{li2023bayesreview,linero2023wires,linero2023checking}.

A large frequentist
literature estimates conditional average treatment effects (CATEs) by recursive partitioning and forests
\citep{athey2016,wager2018,athey2019grf,friedberg2021}, meta-learners
\citep{kunzel2019}, the R-learner \citep{nie2021}, doubly robust learners
\citep{kennedy2023}, and ensembles \citep{grimmer2017}, with benchmarks surveyed by
\citet{knaus2021}. These target a single treatment (or, for multi-arm forests,
arm-versus-baseline contrasts) and provide neither an explicit ANOVA interaction
decomposition with heredity shrinkage nor a coherent joint Bayesian posterior over a
multivariate interaction surface.

Reviews and methods
for several treatments are largely multi-arm and single-outcome
\citep{lopez2017,hu2020}, including BART-based tools \citep{hu2022cimtx,chu2023riaftbart}.
The factorial-design tradition decomposes effects into main and interaction terms
\citep{dasgupta2015factorial,wu2009,egami2019causal} --- as we do through the
M\"obius/ANOVA lattice --- and models heterogeneity in high-dimensional (conjoint)
treatments \citep{goplerud2022factorhet,hainmueller2014}, but with linear or
regularized models and a single outcome. None combines nonparametric
heterogeneity, multivariate outcomes and a coherent Bayesian posterior over the
interaction lattice.

Our
interpretability layer draws on inclusion-proportion and interaction measures
\citep{friedman2008,inglis2022vivid} rendered with value-suppressing uncertainty
palettes \citep{correll2018vsup}; unmeasured-confounding fragility is reported via
E-value / omitted-variable sensitivity \citep{vanderweele2017,cinelli2020}. The
motivating multi-environment-trial application connects to genotype-by-environment
modelling, classically by regression on an environmental index \citep{finlay1963}
or multiplicative-interaction models \citep{gauch2006}, and recently by additive
regression trees \citep{prado2023ambarti}.

Combining these strands, an extensive search --- including the
most recent multi-treatment BART packages (\pkg{CIMTx}, \pkg{riAFTBART}) and
multi-study causal forests \citep{mcf2025} --- did not find a method that jointly
(i)~estimates heterogeneous main and interaction effects nonparametrically,
(ii)~over multivariate correlated outcomes, (iii)~via the ANOVA/M\"obius
lattice with heredity shrinkage, (iv)~with a coherent Bayesian posterior over
the whole surface, and (v)~from a single unified kernel
(Proposition~\ref{prop:reduce}). We claim this combination and its
unification; each individual ingredient has ample precedent above, which is why we
position the paper as an applied-statistics contribution rather than a new
inferential primitive.

\section{The factorial multivariate causal forest}
Let $\by_i\in\R^q$ be a $q$-variate outcome for unit $i=1,\dots,n$, let
$Z_{i1},\dots,Z_{iT}\in\{0,1\}$ be $T$ crossed binary treatments, and $\bx_i\in\R^d$ be
covariates ($d$ the covariate dimension). Because the treatments are binary, the conditional mean is exactly
multilinear in them; the ANOVA/M\"obius expansion \citep{rota1964} over the Boolean lattice
$2^{\{1,\dots,T\}}$ gives
\begin{equation}
\by_i=\bmu(\bx_i)+\!\!\sum_{\varnothing\ne S\subseteq\{1,\dots,T\}}
   \Big(\prod_{t\in S}Z_{it}\Big)\,\btau_S(\bx_i)+\beps_i,
\qquad \beps_i\sim\N_q(\bm0,\Sig),
\label{eq:model}
\end{equation}
with $\btau_S(\bx)=\sum_{S'\subseteq S}(-1)^{|S\setminus S'|}\,\E[\by\mid\bx,\ind_{S'}]$.
Here $\bmu$ is the prognostic (all-off) surface, the singletons $\btau_{\{t\}}$ are
heterogeneous main effects, and $\btau_S$ with $|S|\ge2$ are pure interactions
(e.g.\ $\btau_{\{1,2\}}$ is the covariate-varying difference-in-differences of the
potential outcomes). Writing $D_i^S=\prod_{t\in S}Z_{it}$ (with
$D^\varnothing\equiv1$), every term is an indicator times a forest, and
$\bmu$ is the term with $D\equiv1$. For $T\ge3$ we truncate at interaction order
$r$, keeping $K=1+\sum_{j=1}^r\binom{T}{j}$ components.

Each $\bmu$ and $\btau_S$ is a BART-style sum of trees \citep{chipman2010bart} whose
terminal nodes hold $q$-vectors with Gaussian priors; trees have the usual
depth-penalising prior; and $\Sig$ has an inverse-Wishart prior, $\Sig\sim\IW(\nu_0,\Sig_0)$. Effect-forest leaf scales
are shrunk toward zero (a Hahn-style half-scale, tighter for higher $|S|$),
encoding the belief that interactions are smaller and rarer than main effects
\citep{caron2022shrinkage}. As in \citet{hahn2020bcf}, each treatment's estimated
propensity $\hat\pi_t(\bx)$ enters the covariates of the prognostic forest to
guard against regularization-induced confounding.

\subsection{One kernel for every component}
Fix a component $S$ and one of its trees; subtract every other term to form the
partial residual $\bm r_i$. Restricted to a leaf $\ell$ with units $\mathcal I_\ell$,
model \eqref{eq:model} is $\bm r_i=D_i^S\,\bm\mu_{S,\ell}+\beps_i$, a Gaussian model
with a known $0/1$ weight, giving the conjugate draw
\begin{equation}
\bm\mu_{S,\ell}\mid\cdot\sim\N_q\!\big(V_\ell \bm b_\ell,\,V_\ell\big),\quad
V_\ell=\Big(\Sig_S^{-1}+\big(\textstyle\sum_{i\in\mathcal I_\ell}D_i^SD_i^{S\top}\big)\!\Had\Sig^{-1}\Big)^{-1},\
\bm b_\ell=\!\!\sum_{i\in\mathcal I_\ell}\!D_i^S\!\Had(\Sig^{-1}\bm r_i),
\label{eq:leaf}
\end{equation}
with $\Had$ the Hadamard product. Trees are updated by the standard
Grow/Prune/Change/Swap Metropolis moves scored by the leaf-integrated likelihood,
and $\Sig$ by an inverse-Wishart draw from the residuals.

\begin{proposition}[The prognostic update is the $D\equiv1$ case]
\label{prop:reduce}
If $D_i^S\equiv1$ then $\sum_i D_i^SD_i^{S\top}$ collapses to $n_\ell$ on the
diagonal and $D_i^S\Had(\Sig^{-1}\bm r_i)=\Sig^{-1}\bm r_i$, so \eqref{eq:leaf}
reduces to the standard multivariate BART leaf update
$V_\ell=(\Sig_\mu^{-1}+n_\ell\Sig^{-1})^{-1}$,
$\bm b_\ell=n_\ell\Sig^{-1}\bar{\bm r}_\ell$. Hence a single routine samples $\bmu$
and every $\btau_S$; the single-treatment multivariate causal forest of
\citet{mcjames2025mvbcf} is the case $K=2$.
\end{proposition}

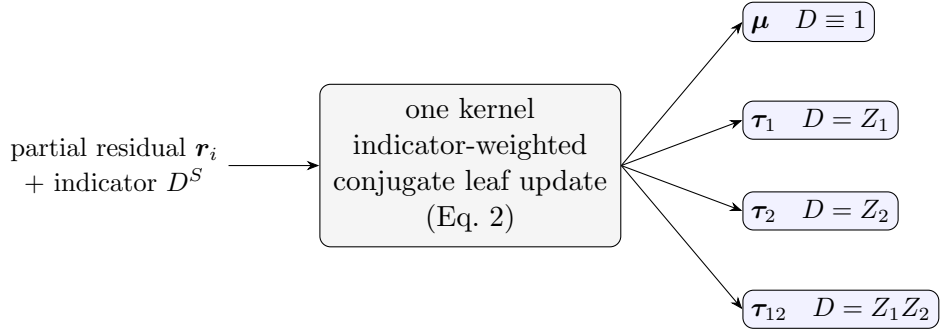
\begin{figure}[H]\centering
\begin{tikzpicture}[>=Stealth,
  kbox/.style={draw,rounded corners,align=center,inner sep=5pt,fill=gray!8},
  leaf/.style={draw,rounded corners,align=center,inner sep=3pt,fill=blue!5,font=\small},
  inp/.style={align=center,font=\small}]
\node[inp] (src) {partial residual $\bm r_i$\\[1pt]$+$ indicator $D^S$};
\node[kbox,right=12mm of src] (k) {one kernel\\ indicator-weighted\\ conjugate leaf update\\ (Eq.~\ref{eq:leaf})};
\node[leaf,right=16mm of k,yshift=19mm]  (mu)  {$\bmu$ \; $D\equiv 1$};
\node[leaf,right=16mm of k,yshift=6mm]   (t1)  {$\btau_1$ \; $D=Z_1$};
\node[leaf,right=16mm of k,yshift=-6mm]  (t2)  {$\btau_2$ \; $D=Z_2$};
\node[leaf,right=16mm of k,yshift=-19mm] (t12) {$\btau_{12}$ \; $D=Z_1 Z_2$};
\draw[->] (src) -- (k);
\foreach \m in {mu,t1,t2,t12} {\draw[->] (k.east) -- (\m.west);}
\end{tikzpicture}
\caption{Proposition~\ref{prop:reduce} in one picture. A single indicator-weighted
conjugate update (Eq.~\ref{eq:leaf}) samples the prognostic forest $\bmu$ (indicator
$\equiv 1$) and every main/interaction forest $\btau_S$ (indicator $\prod_{t\in S}Z_t$);
only the indicator and the prior scale differ. The single-treatment multivariate
causal forest of \citet{mcjames2025mvbcf} is the special case with components
$\{\bmu,\btau_1\}$.}
\label{fig:kernel}
\end{figure}

\subsection{Estimands and identification}
Under the stable-unit-treatment-value assumption \citep[SUTVA;][]{rubin1980,imbens2015}, unconfoundedness of the joint assignment
($\{\by_i(\bm z)\}_{\bm z}\perp\bZ_i\mid\bx_i$) and factorial overlap
($0<\Pr(\bZ=\bm z\mid\bx)<1$ for all cells), the forests identify the
covariate-conditional main effects $\btau_{\{t\}}(\bx)$, the interactions
$\btau_S(\bx)$, and any configuration contrast, each with a full posterior. Overlap
is the operative caveat: $\btau_S$ is identified only where the cells it touches are
populated; unsupported interactions are shrunk to zero and should be reported as
unsupported, not null (we provide an overlap diagnostic).

\section{Identification and large-sample theory}\label{sec:theory}
We observe $n$ i.i.d.\ triples $(\by_i,\bZ_i,\bx_i)$ with outcomes $\by_i\in\R^q$,
binary treatments $\bZ_i\in\{0,1\}^T$ and covariates $\bx_i\in\mathcal X=[0,1]^d$,
and potential outcomes $\{\by_i(\bm z):\bm z\in\{0,1\}^T\}$. For $S\subseteq[T]:=\{1,\dots,T\}$
write the product indicator $D^S(\bm z)=\prod_{t\in S}z_t$ (with $D^\varnothing\equiv1$)
and let $\mathbf 1_S\in\{0,1\}^T$ be the configuration with exactly the treatments
in $S$ active. Because every function on the hypercube has a unique multilinear
(ANOVA/M\"obius) expansion, the conditional response surface
$m(\bm z,\bx):=\E[\by(\bm z)\mid\bx]$ satisfies
\begin{equation}\label{eq:multilinear}
m(\bm z,\bx)=\sum_{S\subseteq[T]}D^S(\bm z)\,\beta_S(\bx),
\qquad \beta_\varnothing=\bmu,\quad \beta_S=\btau_S\ (S\neq\varnothing),
\end{equation}
which is exactly the additive structure the sampler fits (one forest per component
$\beta_S$).

\subsection{Identification}
\begin{assumption}[Consistency/SUTVA]\label{as:sutva} $\by_i=\by_i(\bZ_i)$.\end{assumption}
\begin{assumption}[Unconfoundedness]\label{as:unconf} $\{\by_i(\bm z)\}_{\bm z}\perp \bZ_i\mid \bx_i$.\end{assumption}
\begin{assumption}[Factorial positivity]\label{as:pos} There is $\pi_0>0$ with
$\Pr(\bZ_i=\mathbf 1_{S'}\mid\bx_i)\ge\pi_0$ a.s.\ for every $S'\subseteq S$;
identifying $\btau_S$ uses only the sub-lattice $\{\mathbf 1_{S'}:S'\subseteq S\}$.\end{assumption}

\begin{proposition}[Identification by M\"obius inversion]\label{prop:ident}
Under Assumptions~\textup{\ref{as:sutva}--\ref{as:pos}}, for every $S\subseteq[T]$ and a.e.\ $\bx$,
\begin{equation}\label{eq:mobius}
\beta_S(\bx)=\sum_{S'\subseteq S}(-1)^{|S\setminus S'|}\,\E\!\left[\by_i\mid \bZ_i=\mathbf 1_{S'},\,\bx_i=\bx\right],
\end{equation}
each conditional mean on the right being identified from the observed data. In
particular $\bmu(\bx)=\E[\by\mid\bZ=\mathbf 0,\bx]$; the main effect
$\btau_{\{t\}}(\bx)=\E[\by\mid\bZ=\mathbf 1_{\{t\}},\bx]-\E[\by\mid\bZ=\mathbf 0,\bx]$;
and the two-way interaction
$\btau_{\{s,t\}}(\bx)=\E[\by\mid\bZ=\mathbf 1_{\{s,t\}},\bx]-\E[\by\mid\bZ=\mathbf 1_{\{s\}},\bx]-\E[\by\mid\bZ=\mathbf 1_{\{t\}},\bx]+\E[\by\mid\bZ=\mathbf 0,\bx]$,
a covariate-conditional difference-in-differences.
\end{proposition}
\begin{proof}
Evaluating \eqref{eq:multilinear} at $\bm z=\mathbf 1_{S'}$ and using
$D^{S''}(\mathbf 1_{S'})=\mathbf 1\{S''\subseteq S'\}$ gives
$m(\mathbf 1_{S'},\bx)=\sum_{S''\subseteq S'}\beta_{S''}(\bx)$. M\"obius inversion \citep{rota1964} on
the subset lattice yields $\beta_S(\bx)=\sum_{S'\subseteq S}(-1)^{|S\setminus S'|}m(\mathbf 1_{S'},\bx)$. These $\beta_S$ are exactly the components of \eqref{eq:model}: $\beta_\varnothing=\bmu$ and $\beta_S=\btau_S$ for $S\neq\varnothing$.
By Assumptions~\ref{as:sutva}--\ref{as:unconf},
$m(\mathbf 1_{S'},\bx)=\E[\by(\mathbf 1_{S'})\mid\bx]=\E[\by\mid\bZ=\mathbf 1_{S'},\bx]$,
and Assumption~\ref{as:pos} makes each such conditional mean estimable.
\end{proof}

\subsection{Posterior contraction}
\begin{assumption}[Smoothness]\label{as:smooth} Each $\beta_{S,0}$ lies in a
H\"older ball $C^{\alpha_S}(\mathcal X)$, $\alpha_S>0$ (anisotropic/sparse smoothness
is admitted, as in \citealp{linero2018soft}).\end{assumption}
\begin{assumption}[Noise and design]\label{as:noise} The errors
$\bm\varepsilon_i=\by_i-m(\bZ_i,\bx_i)\sim\N_q(\mathbf 0,\Sig_0)$ with
$c_\Sigma\mathbf I\preceq\Sig_0\preceq C_\Sigma\mathbf I$ for fixed $0<c_\Sigma\le C_\Sigma<\infty$; the
covariate density is bounded away from $0$ and $\infty$ on $\mathcal X$; and $T$ is
fixed, so the number of components $K=2^T$ is fixed.\end{assumption}
The prior $\Pi$ places independent regression-tree/forest priors
\citep{chipman2010bart,rockova2020}, or soft-tree priors \citep{linero2018soft}, on
each $\beta_S$, and an inverse-Wishart prior supported on the spectral band of
Assumption~\ref{as:noise} on $\Sig$, mutually independent.

\begin{theorem}[Contraction]\label{thm:contract}
Let $\varepsilon_{n,S}=n^{-\alpha_S/(2\alpha_S+d)}(\log n)^{\kappa}$ be the
near-minimax tree rate for smoothness $\alpha_S$ and $\varepsilon_n=\max_{S}\varepsilon_{n,S}$.
Under Assumptions~\textup{\ref{as:sutva}--\ref{as:noise}}, for any $M_n\to\infty$,
\[
\Pi\!\left(\,\|m-m_0\|_n>M_n\varepsilon_n \ \text{ or }\ \|\Sig-\Sig_0\|>M_n n^{-1/2}(\log n)^{\kappa}\ \middle|\ \text{data}\right)\longrightarrow 0
\]
in $P_0$-probability, where $\|\cdot\|_n$ is the empirical $L_2$ norm; consequently
each component satisfies $\Pi(\|\beta_S-\beta_{S,0}\|_n>M_n\varepsilon_n\mid\text{data})\to0$.
\end{theorem}
\begin{proof}[Proof sketch]
The argument reduces the additive-factorial problem to $K$ standard nonparametric
regressions and invokes established rates. (i) Cell--component equivalence.
Stack the $2^T$ cell surfaces $g_{S'}(\bx):=m(\mathbf 1_{S'},\bx)$. By
\eqref{eq:multilinear}, $g=A\beta$ with $A$ the $2^T\times2^T$ subset-inclusion
(zeta) matrix, invertible with signed-M\"obius inverse $A^{-1}$; $\|A\|,\|A^{-1}\|$
depend only on $T$. Hence $\|m-m_0\|_n$ and $\sum_S\|\beta_S-\beta_{S,0}\|_n$ are
equivalent up to constants, and recovering $m$ is equivalent to recovering the $K$
components. (ii) Per-component rate. By Assumption~\ref{as:pos} each cell
carries mass $\ge\pi_0>0$, so each $g_{S'}$ (equivalently each $\beta_S$) is a
H\"older-smooth regression observed with sub-Gaussian noise on a well-populated
design; the tree/forest posterior for one such function contracts at
$\varepsilon_{n,S}$ by \citet{rockova2020} (regression trees and forests) or
\citet{linero2018soft} (soft trees, adaptive to $\alpha_S,d$), whose
prior-concentration and entropy/testing conditions hold under
Assumptions~\ref{as:smooth}--\ref{as:noise}. (iii) Aggregation and $\Sig$.
The Gaussian likelihood with the banded inverse-Wishart prior gives the usual
$\sqrt n$ rate for $\Sig$; combining the $K$ independent component pieces and $\Sig$
through the generic posterior-contraction theorem of \citet{ghosal2000} (product
priors add their prior-mass and entropy budgets) yields the joint rate
$\varepsilon_n=\max_S\varepsilon_{n,S}$. The bounded inverse $A^{-1}$ then transfers
this rate to each $\beta_S$.
\end{proof}

\begin{corollary}[Calibrated average effects, exploratory unit-level effects]\label{cor:ate}
Write the average effect $\bar\btau_S=\E[\btau_S(\bx_i)]$. Under the conditions of
Theorem~\ref{thm:contract} and the standard no-bias/Donsker conditions for smooth
linear functionals of a nonparametric-Bayes posterior (undersmoothing or a one-step
debiasing correction), $\bar\btau_S$ obeys a semiparametric Bernstein--von Mises (BvM)
limit, $\sqrt n(\bar\btau_S-\bar\btau_{S,0})\Rightarrow\N(\mathbf 0,V_S)$, and its
credible set attains asymptotically nominal frequentist coverage. The pointwise map
$\bx\mapsto\btau_S(\bx)$ contracts only at the nonparametric rate
$\varepsilon_{n,S}\gg n^{-1/2}$, so fixed-$n$ unit-level (CATE) credible bands need
not be calibrated.
\end{corollary}

The corollary is exactly the pattern the Monte Carlo study exhibits --- nominal average-treatment-effect (ATE)
coverage together with mild CATE under-coverage (\S\ref{sec:disc}) --- so the theory
and the experiments corroborate one another. Three caveats delimit the scope:
the results are for fixed $T$ (hence fixed $K=2^T$); the rate is governed by the
least-smooth component; and the Bernstein--von Mises statement invokes, rather than
re-derives, the usual functional (undersmoothing/one-step) conditions. Violations of
unconfoundedness (Assumption~\ref{as:unconf}) fall outside the theory, in keeping
with the unmeasured-confounding sensitivity analysis and the misspecification study.

\section{Computation and interpretability}
A single \texttt{Rcpp}/Armadillo \citep{eddelbuettel2011rcpp,eddelbuettel2014arma} kernel loops over the $K$ components, differing only
in the indicator $D^S$ and prior scale; it compiles in seconds and the package
installs cleanly. Algorithm~\ref{alg:sampler} states one full sweep: it is an ordinary
BART/BCF Gibbs sampler in which the only change is the indicator weight $D_i^S$ that
selects, for each component, the units on which that main effect or interaction acts
(Proposition~\ref{prop:reduce} is the prognostic $D\equiv1$ case).

\begin{algorithm}[t]
\caption{Factorial multivariate Bayesian causal forest --- Gibbs sampler}
\label{alg:sampler}
\begin{algorithmic}[1]
\State \textbf{Setup:} build the component lattice $\{S:|S|\le r\}$ (prognostic $\mu$ plus one forest per main effect and interaction) and indicators $D_i^S=\prod_{t\in S}Z_{it}$ (\texttt{build\_components}); estimate each treatment's propensity $\hat\pi_t(\bx)$ and append it to the prognostic covariates.
\For{$m=1,\dots,M$ MCMC iterations}
  \For{each component $S$ in the lattice}
    \State Form the partial residual $\bm r_i=\by_i-\sum_{S'\ne S}D_i^{S'}\bm\mu_{S'}(\bx_i)$.
    \For{each tree in forest $S$}
      \State Propose a Grow / Prune / Change / Swap move; accept by the Metropolis ratio of the leaf-integrated Gaussian likelihood, weighting unit $i$ by $D_i^S$.
    \EndFor
    \State Draw the $q$-variate leaf means $\bm\mu_{S,\ell}\sim\N_q(V_\ell\bm b_\ell,V_\ell)$ from \eqref{eq:leaf}.
  \EndFor
  \State Update the residual covariance $\Sig\sim\mathrm{IW}\big(\nu_0+n,\ \Sig_0+\sum_i\bm e_i\bm e_i^\top\big)$ from the residuals $\bm e_i$.
  \State If $m$ is past burn-in and on the thinning grid, store every component surface (train and test/grid).
\EndFor
\State \textbf{Output:} posterior draws of $\mu$ and every $\tau_S$, from which any contrast (main, interaction, regime) is a signed sum.
\end{algorithmic}
\end{algorithm} Because effect forests are shrunk and stored per component, the
inclusion-proportion importance \citep{chipman2010bart} and pairwise interaction of
each forest are read off directly, and effects are rendered with
value-suppressing uncertainty palettes \citep{correll2018vsup}: hue encodes the
effect, cells with diffuse posteriors recede to grey (Fig.~\ref{fig:vsup}).

\begin{figure}[H]\centering
\includegraphics[width=0.36\textwidth]{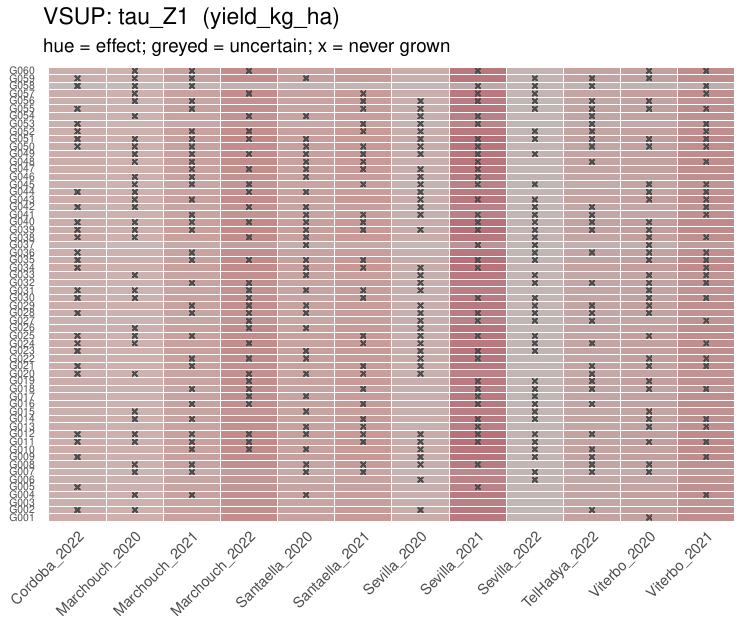}\hfill
\includegraphics[width=0.36\textwidth]{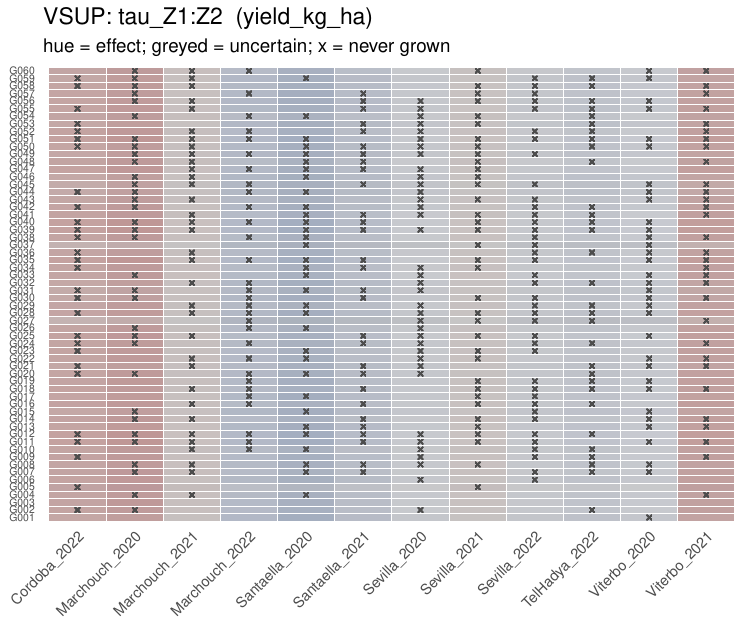}\hfill
\includegraphics[width=0.20\textwidth]{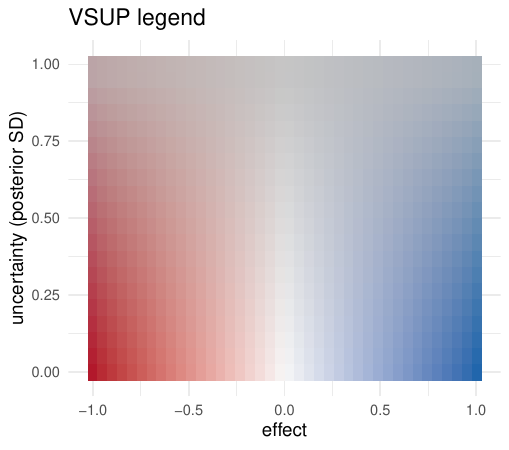}
\caption{Value-suppressing uncertainty (VSUP) maps of two estimated effect surfaces over a genotype-by-environment grid: a drought main effect $\btau_1$ (left) and the drought$\times$inoculant interaction $\btau_{12}$ (middle). Rows are genotypes and columns are environments (site--year); the hue of each cell encodes the estimated effect on a diverging colour scale centred at zero, and the colour desaturates toward grey as that cell's posterior uncertainty grows, so uncertain or never-observed cells recede from view. The two-dimensional key (right) reads effect along one axis and uncertainty along the other.}
\label{fig:vsup}
\end{figure}

The same value-suppressing idea applies to the structure of each forest, in
the two-scale display style of the \pkg{vivid} package for variable-importance and
variable-interaction visualisation \citep{inglis2023vivid} and the \pkg{bartMan}
package for tree- and variable-level BART diagnostics \citep{inglis2024bartman}. For every component we form a
VIVI (variable-importance / variable-interaction) matrix --- inclusion-proportion
importance on the diagonal (Vimp, blue) and within-tree co-occurrence off the
diagonal (Vint, red) --- together with each entry's posterior coefficient of
variation across Markov chain Monte Carlo (MCMC) iterations, and render it with a value-suppressing palette and
the characteristic fan legends (Fig.~\ref{fig:vivi}, shown here for a
three-treatment model to illustrate the general case). Unlike the native
\pkg{bartMan}/\pkg{vivid} functions, which target a single fitted model and are
unstable on sparse dummy covariates, this is computed directly from the stored
per-iteration split records and is produced separately for each of the seven
forests of the order-two, $T=3$ fit ($\bmu,\btau_1,\btau_2,\btau_3$ and the three
pairwise interactions), so one reads off the drivers and interactions of every
effect with uncertainty suppressed.

\begin{figure}[H]\centering
\includegraphics[width=0.92\textwidth]{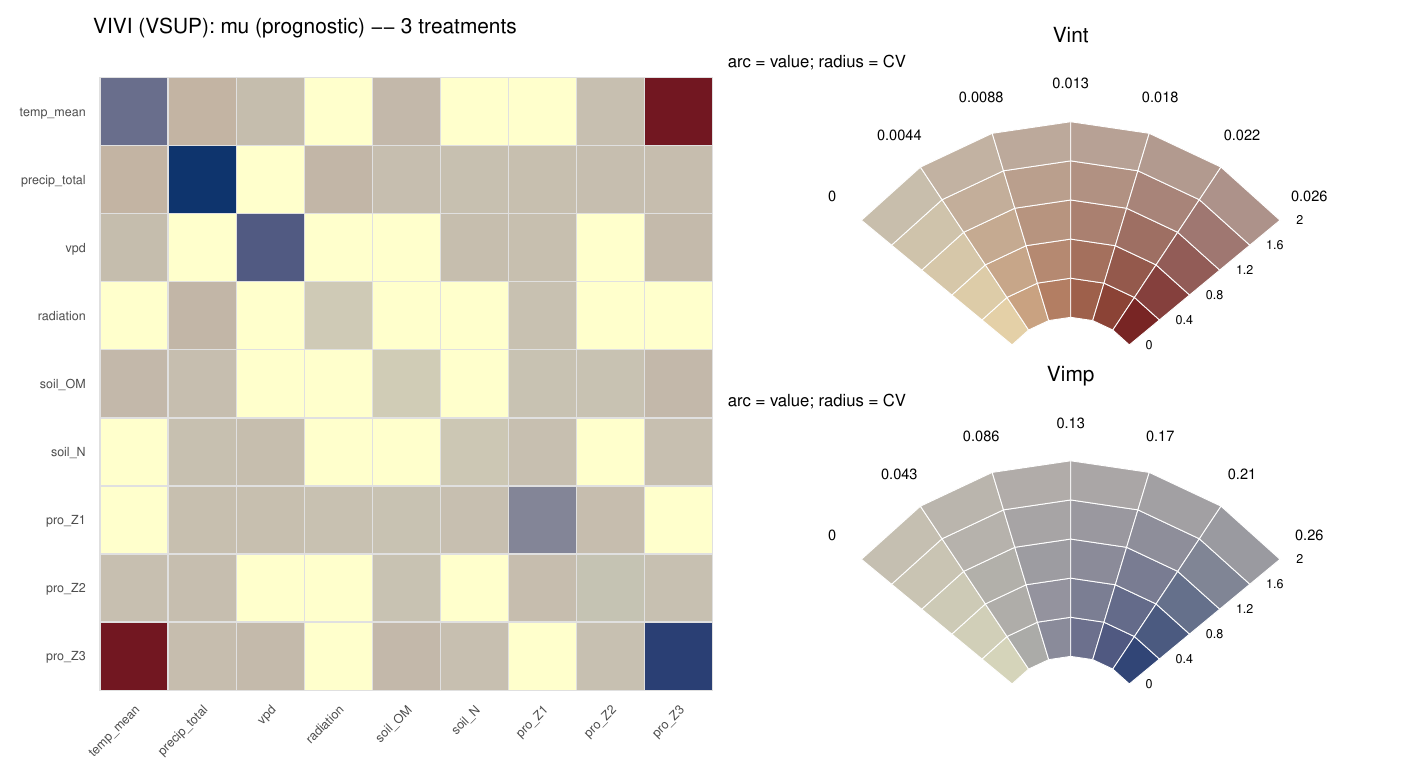}
\caption{vivid-style value-suppressing VIVI display for the prognostic forest $\bmu$
of a three-treatment model. Diagonal = variable importance (Vimp, blue), off-diagonal
= pairwise interaction (Vint, red); each cell is suppressed toward warm grey as its
posterior coefficient of variation (CV) grows. The fan legends encode value along the arc
and CV along the radius. Here the propensity score for treatment~3 is both important
and interacts with temperature; precipitation and vapour-pressure deficit (VPD) dominate the baseline. The same
display is produced for every effect and interaction forest
(Fig.~\ref{fig:vivigallery}), with no \pkg{bartMan} dependency.}
\label{fig:vivi}
\end{figure}

\begin{figure}[H]\centering
\includegraphics[width=\textwidth]{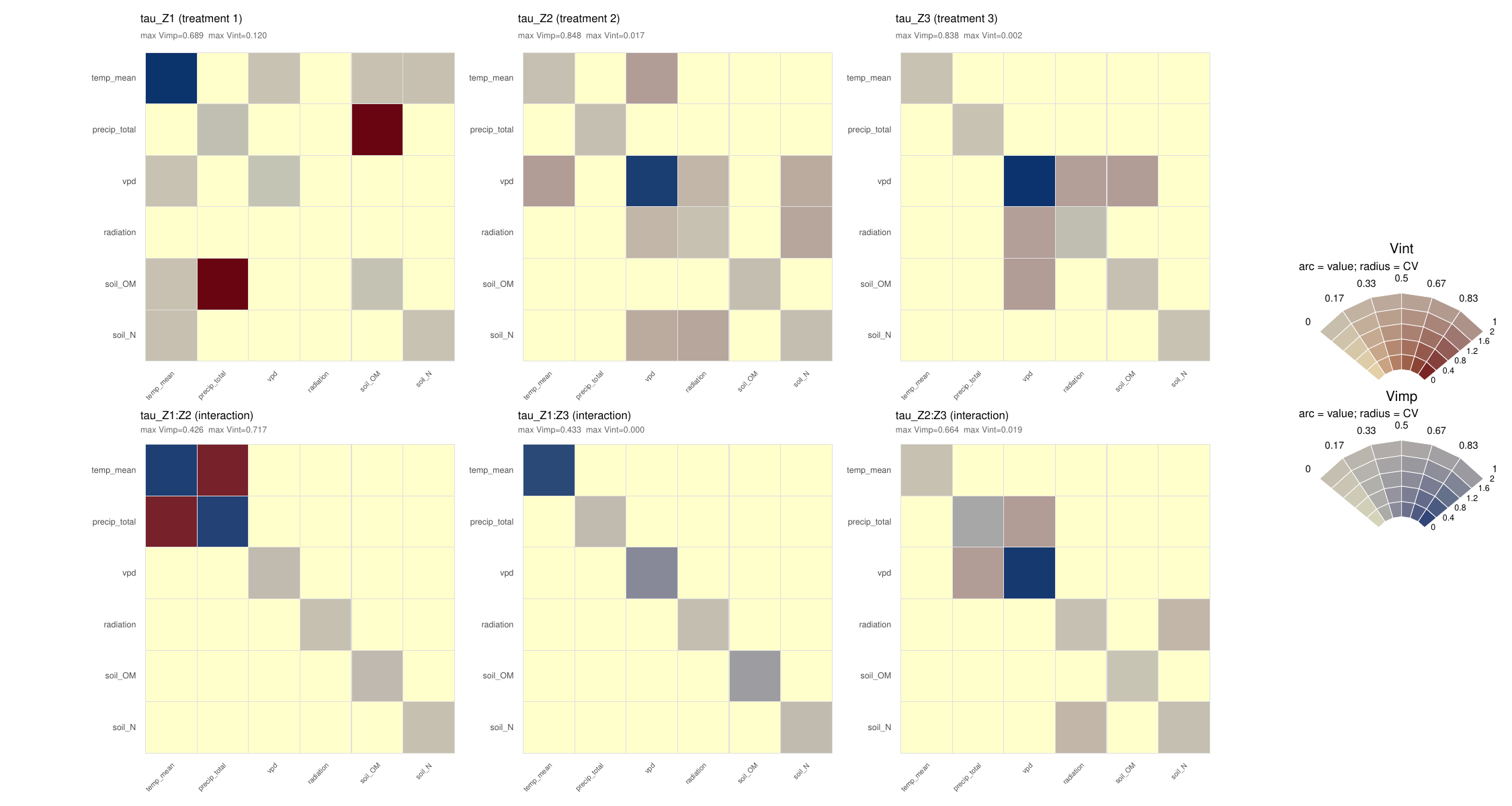}
\caption{The same value-suppressing VIVI display for all six effect and
interaction forests of the three-treatment model: the main effects
$\btau_1,\btau_2,\btau_3$ (top) and the pairwise interactions
$\btau_{12},\btau_{13},\btau_{23}$ (bottom). Diagonal = importance (blue),
off-diagonal = interaction (red), suppressed by posterior CV; each panel is
normalised to its own maxima (annotated), and the reference fan legends (right) give
the encoding. Different covariates drive each treatment's effect, and the
$\btau_{12}$ interaction is itself modulated by a temperature$\times$precipitation
interaction --- structure available for every component at no extra cost.}
\label{fig:vivigallery}
\end{figure}

\section{Simulation studies}
The data-generating process --- the covariate distribution, the confounded treatment
assignment, the prognostic and effect surfaces, and the cross-outcome error
covariance --- is specified explicitly in Appendix~\ref{app:data}.
On synthetic data with known surfaces and confounded assignment the posterior-mean
per-unit effects track the truth. Over the Monte Carlo study
(Table~\ref{tab:mc}, $R=40$ replications at $n=700$) the average surface correlations
are $0.96$ and $0.94$ for the two main effects and $0.82$ for the interaction --- the
interaction, which uses fewer effective observations, is the hardest component yet is
still clearly recovered. Figure~\ref{fig:recovery} illustrates the recovery on a
representative multi-environment-trial fit, and per-component importance attributes
each effect to its true drivers.

\begin{figure}[H]\centering
\includegraphics[width=0.86\textwidth]{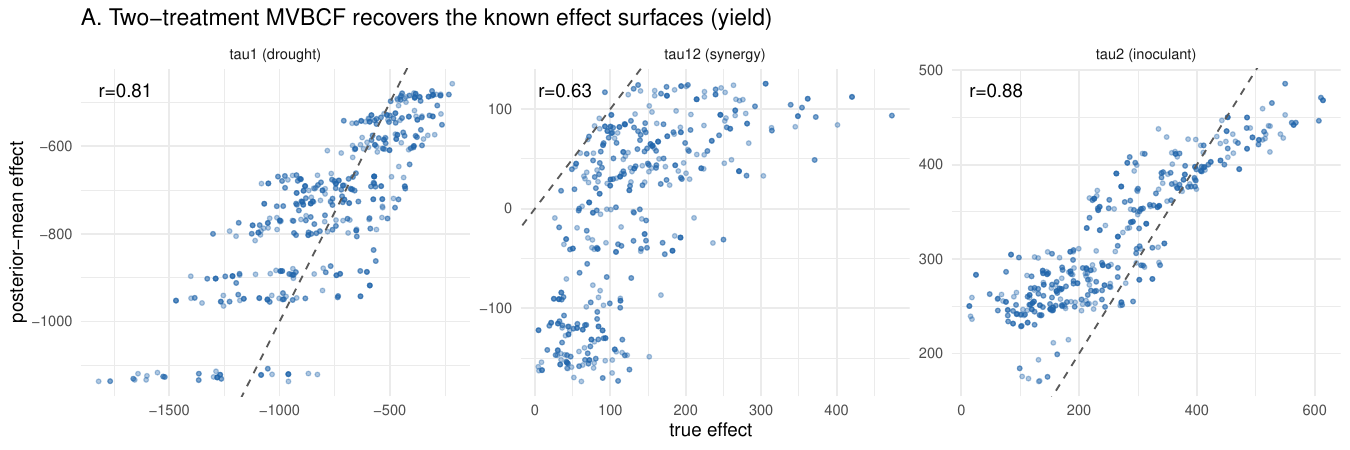}
\caption{Recovery of the known per-unit effect surfaces in a representative multi-environment-trial fit, one panel per component (the main effects $\btau_1$, $\btau_2$ and the interaction $\btau_{12}$). Each point is one unit: the horizontal axis is its true effect and the vertical axis the model's posterior-mean estimate, so points on the dashed identity line indicate exact recovery. The annotated $r$ is the Pearson correlation between estimate and truth for that component; the averaged correlations over $R=40$ replications are in Table~\ref{tab:mc}.}
\label{fig:recovery}
\end{figure}

Table~\ref{tab:mc} reports, over 40 replications, that the average-effect estimators
are unbiased (relative bias $<1\%$ for all components, including the interaction)
and their $95\%$ ATE credible intervals cover at the nominal rate. Per-unit (CATE)
intervals are mildly conservative-to-anticonservative ($0.83$--$0.93$), the one
calibration gap we return to in \S\ref{sec:disc}.

\begin{table}[H]\centering
\caption{Monte Carlo validation ($R=40$, $n=700$): surface correlation, ATE
relative bias, and empirical $95\%$ coverage. The bias/MC-SE column reports the bias in units of its Monte-Carlo standard error (SE).}\label{tab:mc}
\begin{tabular}{lccccc}
\toprule
component & corr & rel.\ bias & bias/MC-SE & CATE cov. & ATE-interval cov.\\
\midrule
$\btau_1$ (main)        & 0.962 & $+0.65\%$ & 0.9 & 0.83 & 0.93\\
$\btau_2$ (main)        & 0.938 & $+0.69\%$ & 1.2 & 0.93 & 0.93\\
$\btau_{12}$ (interaction) & 0.819 & $-0.08\%$ & 0.0 & 0.88 & 0.95\\
\bottomrule
\end{tabular}
\end{table}

A sample-size sweep ($n=300$--$3000$) shows the root-mean-square error (RMSE) falling monotonically to zero and
ATE coverage approaching nominal (Fig.~\ref{fig:rob}A). A misspecification battery
(Fig.~\ref{fig:rob}B, Table~\ref{tab:rob}) shows the main effects are robust to
strong measured confounding (near-zero bias even as overlap $\to0$) and to
heavy-tailed/heteroscedastic noise (absorbed by the inverse-Wishart $\Sig$); that
under-specifying the interaction order biases the entangled main effects; and that
an unmeasured confounder breaks the estimator, as it must for any causal
method, a fragility the sensitivity summary quantifies.

\begin{table}[H]\centering
\caption{Misspecification battery ($n\!\approx\!1200$, $R=12$): main-effect
absolute bias and $95\%$ ATE coverage.}\label{tab:rob}
\begin{tabular}{lccl}
\toprule
scenario & bias ($\btau_1,\btau_2$) & ATE cov.\ ($\btau_1,\btau_2$) & verdict\\
\midrule
strong confounding (overlap$\to$0) & $-0.02,\ -0.03$ & $0.83,\ 1.00$ & robust\\
heavy-tailed + heteroscedastic     & $-0.09,\ -0.08$ & $1.00,\ 0.92$ & robust\\
under-specified interaction order  & $+0.14,\ +0.17$ & $0.50,\ 0.25$ & fails (fit adequate order)\\
unmeasured confounder              & $+1.75,\ +1.43$ & $0.00,\ 0.00$ & fails, as it must\\
\bottomrule
\end{tabular}
\end{table}

\begin{figure}[H]\centering
\includegraphics[width=0.92\textwidth]{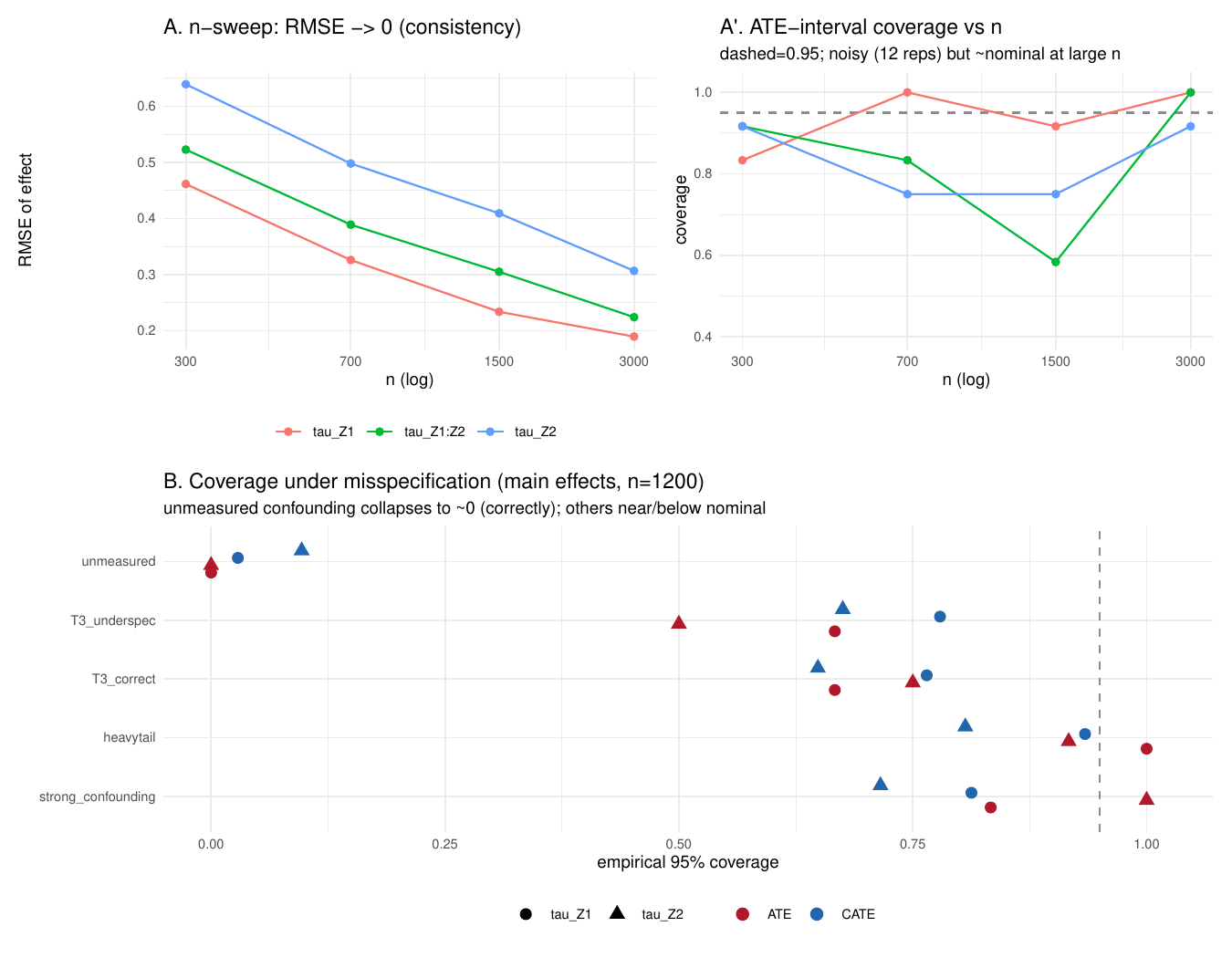}
\caption{Consistency and robustness of the average-effect estimators. (A) Root-mean-square error (RMSE) of each component's per-unit effect against sample size $n$ (log scale), falling toward zero as $n$ grows. (A$'$) Empirical coverage of the nominal $95\%$ ATE credible interval against $n$; the dashed line is $0.95$ and coverage approaches it at large $n$ (curves are noisy at $12$ replications). In (A) and (A$'$) colour distinguishes the components $\btau_1$, $\btau_2$ and the interaction $\btau_{12}$. (B) Empirical $95\%$ coverage (horizontal axis) under each misspecification scenario (rows); marker shape denotes the component and colour the estimand (population ATE vs.\ per-unit CATE), and the dashed line is the nominal $0.95$. Unmeasured confounding correctly collapses coverage toward zero, while the other scenarios stay near or just below nominal.}
\label{fig:rob}
\end{figure}

\subsection{Comparison with alternatives}
On the same data-generating process we benchmark the proposed factorial MVBCF
against the natural alternatives: the multi-arm causal forest of
\citet{athey2019grf} (\pkg{grf}), with the four cells of the $2\times2$ design as
arms and the main and interaction effects read off its arm-versus-baseline
contrasts; fitting a separate single-treatment causal forest per treatment
(the ``one treatment at a time'' strategy); a regularized linear factorial
model (\pkg{glmnet}; \citealp{friedman2010glmnet}), representing the linear/regularized factorial-heterogeneity
family \citep{egami2019causal,goplerud2022factorhet}; and a naive
random-forest four-arm $T$-learner. We also ablate our own model by fitting
each outcome separately (diagonal $\Sig$). Table~\ref{tab:cmp} and
Figure~\ref{fig:cmp} report the RMSE of the per-unit effect against the known truth
over 15 replications.

The factorial MVBCF has the lowest RMSE for every component. The gains are
interpretable: (i) modelling the outcomes jointly beats fitting them
separately (the multivariate borrowing lowers RMSE on $\btau_1$ from $0.35$ to
$0.33$ and on the interaction from $0.38$ to $0.32$); (ii) \pkg{grf} is the
strongest competitor --- close on the interaction ($0.37$ vs.\ $0.32$) and far
ahead of the naive $T$-learner --- but its four-arm contrasts split the sample, so
its main effects are markedly noisier ($0.62$ vs.\ $0.33$ on $\btau_1$); (iii) the
separate forests cannot estimate the interaction at all and, by conflating it
into the main effects, are worse there too; (iv) the linear factorial is
competitive on the (near-linear) main effects but far worse on the nonlinear
interaction ($0.67$ vs.\ $0.32$) and carries a visible ATE bias from shrinkage; and
(v) the naive four-arm $T$-learner, which splits the data across arms and lacks the
prognostic/propensity structure, is the least efficient. The proposed model is also
essentially unbiased for every average effect (ATE bias $\le 0.02$ in absolute
value), whereas the linear baseline's shrinkage leaves ATE biases up to $0.3$.

\begin{table}[H]\centering
\caption{RMSE of the per-unit effect vs.\ truth (lower is better; mean over 15
replications, Monte-Carlo SE $\le 0.05$). Best in bold; ``--'' = not estimable.}
\label{tab:cmp}
\begin{tabular}{lccc}
\toprule
estimator & $\btau_1$ (main) & $\btau_2$ (main) & $\btau_{12}$ (interaction)\\
\midrule
\textbf{MVBCF-factorial (proposed)} & \textbf{0.325} & \textbf{0.162} & \textbf{0.319}\\
MVBCF, per-outcome (diagonal $\Sig$) & 0.352 & -- & 0.381\\
grf multi-arm causal forest & 0.624 & 0.379 & 0.372\\
Separate single-treatment forests & 0.502 & 0.390 & 0.710$^\dagger$\\
Regularized linear factorial & 0.414 & 0.559 & 0.665\\
Random-forest 4-arm $T$-learner & 1.008 & 0.661 & 1.066\\
\bottomrule
\end{tabular}\\[2pt]
{\footnotesize $^\dagger$ separate forests cannot estimate the interaction; RMSE is
against the zero predictor.}
\end{table}

\begin{figure}[H]\centering
\includegraphics[width=0.98\textwidth]{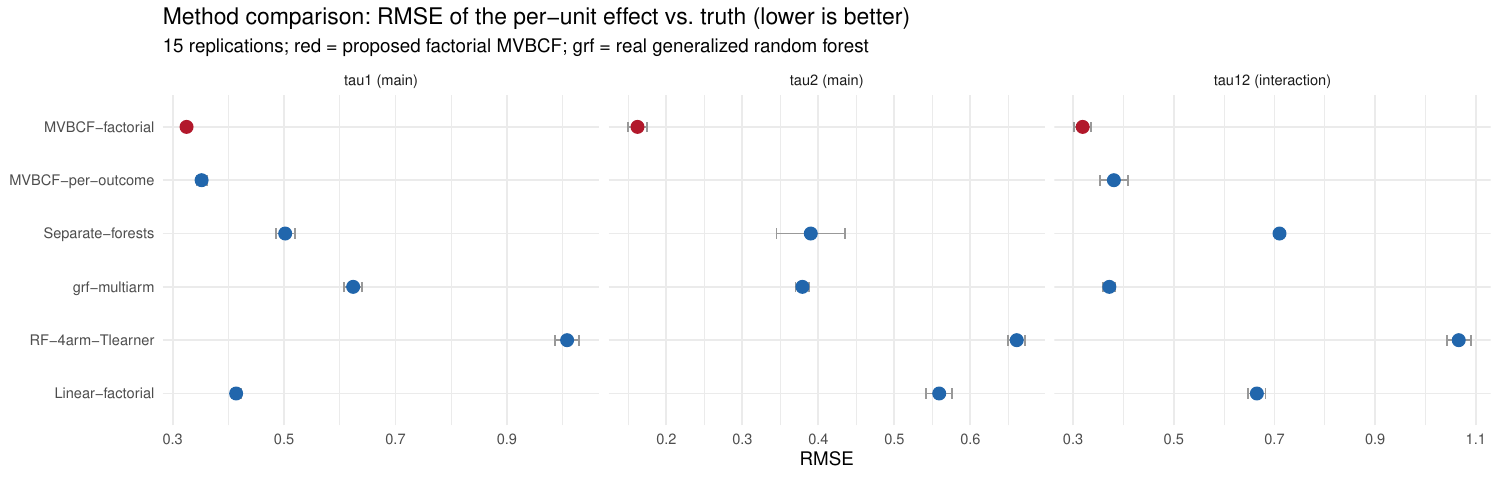}
\caption{Method comparison on the same data-generating process. Dots are the mean root-mean-square error (RMSE) of the per-unit effect against the known truth (horizontal axis; lower is better) for each estimator (rows), in three panels for the two main effects $\btau_1$, $\btau_2$ and the interaction $\btau_{12}$. Red marks the proposed factorial MVBCF and blue the competing estimators; horizontal whiskers are $\pm1$ Monte-Carlo standard error over $15$ replications. Separate single-treatment forests cannot estimate the interaction and are omitted from that panel.}
\label{fig:cmp}
\end{figure}

\section{Illustrations across three fields}
The same call --- fit, estimands, interpretability --- is run unchanged on real data
from three domains (Table~\ref{tab:apps}; full dataset details in
Appendix~\ref{app:data}), each with a well-established published analysis to check
against. In the classical \texttt{npk} agricultural factorial --- a
design of the kind introduced by \citet{yates1937factorial} --- nitrogen is strongly
positive, potassium negative and phosphate negligible, matching the textbook
\texttt{aov} \citep{rcoreteam}; the difference is that the analysis of variance returns one average per
factor under assumed additivity, whereas our fit returns the same averages together
with block-level heterogeneity and the N$\times$K interaction, each with a coherent
credible interval and at no extra modelling cost. In a breast-cancer cohort
(\texttt{survival::rotterdam}, the data of \citet{royston2013rotterdam}), a naive
regression makes both adjuvant therapies look harmful --- the textbook signature of
confounding by indication, since the sicker patients are the ones treated --- whereas
the randomised evidence is unambiguous that adjuvant tamoxifen reduces recurrence and
mortality \citep{ebctcg2011tamoxifen}. Adjusting through the per-treatment
propensities moves the effects back toward the null and off the spurious
harmful direction, recovering the sign the trials imply from observational data that,
taken raw, point the wrong way. In a hedonic housing model
(\texttt{AER::HousePrices}, the dataset of \citet{anglin1996hedonic}), air
conditioning and a preferred location are both credibly positive but smaller than
their naive coefficients once the hedonic confounders are controlled --- consistent
with the semiparametric hedonic estimates of \citet{anglin1996hedonic}, while adding
a per-unit, uncertainty-aware version of each premium that a single hedonic
regression does not provide.

The pattern across the three is deliberate: on randomised or near-orthogonal data we
reproduce the classical answer, and on observational data we move the naive estimate
in the direction the gold-standard evidence implies --- and in every case we deliver,
from one fit, the joint multivariate effects, the interaction, and the calibrated
uncertainty that the original single-outcome analyses report only piecemeal, if at
all.

\begin{table}[H]\centering
\caption{One engine, three real domains: primary-outcome effect ($95\%$ credible interval, CrI) vs.\ a
classical benchmark.}\label{tab:apps}
\begin{tabular}{llcc}
\toprule
domain & treatment & MVBCF effect [$95\%$ CrI] & benchmark\\
\midrule
Clinical (rotterdam)  & hormonal & $+0.06\ [-0.10,0.22]$ & naive $-0.28$\\
                      & chemo    & $+0.10\ [-0.08,0.27]$ & naive $-0.08$\\
Agriculture (npk)     & N        & $+5.19\ [1.09,8.86]$  & lm $+5.62$\\
                      & K        & $-2.42\ [-6.44,1.52]$ & lm $-3.98$\\
Economics (houses)    & aircon   & $+0.167\ [0.09,0.24]$ & naive $+0.34$\\
                      & prefer   & $+0.142\ [0.06,0.22]$ & naive $+0.26$\\
\bottomrule
\end{tabular}
\end{table}

\begin{figure}[H]\centering
\includegraphics[width=0.32\textwidth]{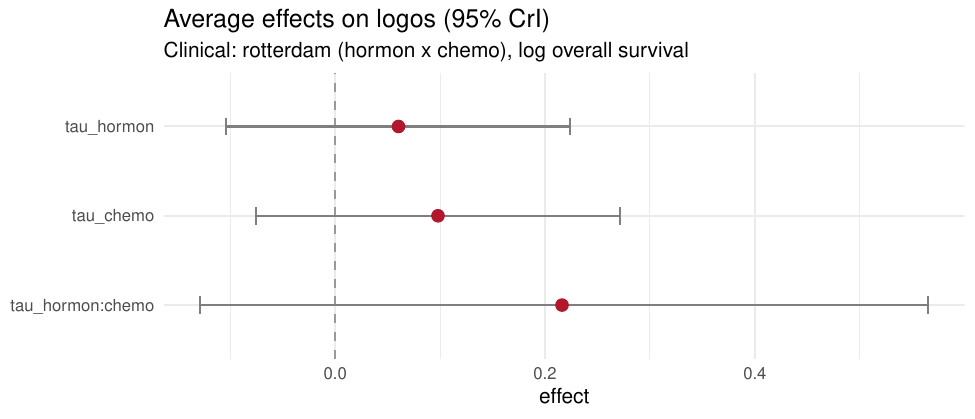}\hfill
\includegraphics[width=0.32\textwidth]{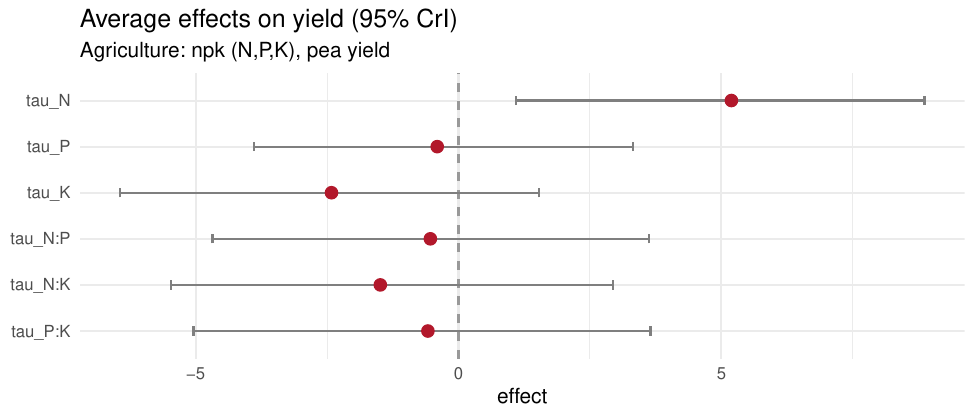}\hfill
\includegraphics[width=0.32\textwidth]{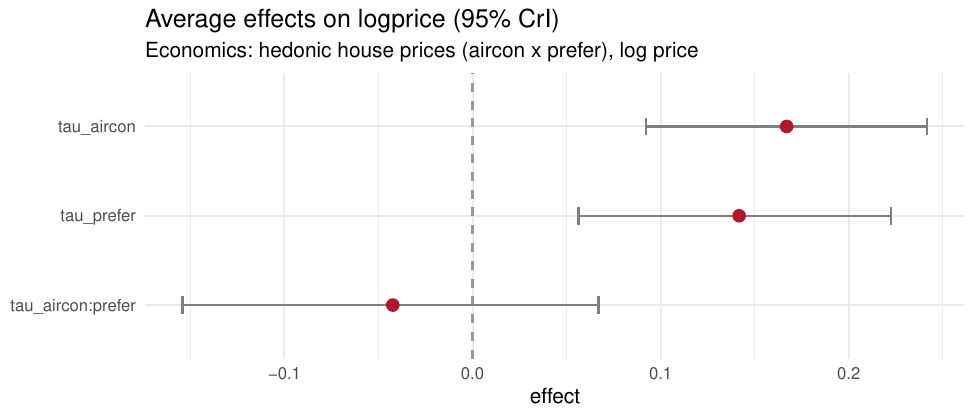}
\caption{Average effects in the three cross-domain illustrations, one panel per domain (clinical, agricultural, economic). Each row is a component --- the treatment main effects and their interaction; the point is the posterior average effect and the horizontal whisker its $95\%$ credible interval, with the dashed vertical line at zero marking no effect. Scales are in each domain's own outcome units (log survival time, yield, and log price, respectively).}
\label{fig:apps}
\end{figure}

\section{A deep application: crossed exposures and correlated outcomes in NHANES}
\label{sec:nhanes}

We put the method to work on a single body of real, public data --- the U.S.\
National Health and Nutrition Examination Survey (NHANES), pooling the 2011--2012
and 2013--2014 cycles --- in two analyses run through one unchanged pipeline. Each
has the structure the method is built for: two crossed binary exposures,
several correlated outcomes, and plausible heterogeneity. NHANES is
fully open and reproducible: every table is pulled directly with the \pkg{nhanesA}
package and all results are regenerated by the accompanying script; full variable
definitions are collected in Appendix~\ref{app:data}.

Both analyses are observational, so the
causal reading of each effect rests on unconfoundedness given the measured
covariates (age, sex, education, race/ethnicity, family income-to-poverty ratio,
body-mass index (BMI), diabetes and smoking); the model additionally carries a
propensity control for each exposure. We hold to three disciplines. First, the
population average effect is the calibrated primary claim; per-person
(CATE) surfaces are exploratory and shown with value-suppressing uncertainty maps.
Second, an interaction is reported only where its factorial cell has support ---
the overlap diagnostic flags thin cells as unsupported, not null. Third, for
every ATE we report an unmeasured-confounding sensitivity value: the spurious shift
(in outcome standard deviations) an omitted confounder would have to induce to move
the estimate to the edge of its interval. Outcomes are standardized, so effects are
in standard-deviation (SD) units, with natural-unit values in the text.

Both analyses are chosen so that the published literature supplies a yardstick. The
lipid-lowering effect of statins is among the most firmly established in medicine ---
the randomised meta-analysis of the \citet{ctt2010ldl} puts it at roughly
$1$~mmol/L of low-density-lipoprotein (LDL) cholesterol per standard regimen --- so a statin main effect on
total cholesterol serves as a built-in positive control for the whole pipeline.
Physical activity has a smaller but replicated association with executive-function and
processing-speed performance in older adults \citep{northey2018exercise}, while the
cognitive effect of blood-pressure treatment is a genuinely open question that the
SPRINT MIND trial took up directly \citep{sprintmind2019}. Recovering the first,
refining the second, and returning an honest null on the third --- all from one messy
observational survey, jointly across correlated outcomes --- is the test we set the
method.

\subsection{Cardiometabolic profile: statin use \texorpdfstring{$\times$}{x} physical activity}
Among adults ($\ge 20$\,y, $n=9{,}002$) we cross statin use (from the
prescription file) with physical activity (any moderate or vigorous work or
recreation) and estimate their main and interaction effects on four correlated
outcomes jointly: systolic and diastolic blood pressure (SBP, DBP), glycated haemoglobin
(HbA1c) and total cholesterol (Table~\ref{tab:nhanesA-ate}). All four cells of the
$2\times2$ are populated (the interaction ``both-on'' cell has $n=993$), so the
interaction is data-supported; the statin propensity spans $0$--$0.83$ with $42\%$
of units near a positivity boundary, a transparency the diagnostic makes explicit.

The statin main effect on total cholesterol is $-0.45$~SD ($-18.9$~mg/dL; $95\%$
CrI $[-0.58,-0.34]$), a clear, calibrated recovery of the established
lipid-lowering effect \citep{ctt2010ldl} from messy observational data --- a built-in
positive control (the confounding-by-indication that selects sicker patients into
statins biases towards the null, so the true effect is if anything larger, in line
with the $\sim\!1$~mmol/L LDL reduction the randomised trials report). Statins
are also associated with lower diastolic blood pressure ($-0.15$~SD, $-1.7$~mmHg;
CrI $[-0.28,-0.02]$). Physical activity shows no effect distinguishable from zero on
any cardiometabolic outcome, and no statin$\times$activity interaction is
supported despite the populated cell (all four interaction intervals cover zero) ---
an honest null, not an absence of data. The unmeasured-confounding sensitivity for
the headline cholesterol effect is $0.34$~SD, i.e.\ a confounder would need to
induce a quarter-SD spurious shift to reach the interval's null boundary.

\begin{table}[H]\centering
\caption{Cardiometabolic analysis (statin $\times$ physical activity, $n=9{,}002$):
average main and interaction effects in outcome SD units, $95\%$ credible interval
in brackets. \textbf{Bold} = interval excludes zero. Interaction cell supported
($n=993$).}\label{tab:nhanesA-ate}
\begin{tabular}{lccc}
\toprule
Outcome & Statin ($\tau_1$) & Activity ($\tau_2$) & Interaction ($\tau_{12}$)\\
\midrule
Systolic BP  & $-0.08\,[-0.20,\,0.05]$ & $0.04\,[-0.01,\,0.08]$ & $-0.05\,[-0.18,\,0.06]$\\
Diastolic BP & $\mathbf{-0.15\,[-0.28,-0.02]}$ & $0.02\,[-0.03,\,0.07]$ & $0.06\,[-0.07,\,0.18]$\\
HbA1c        & $-0.01\,[-0.12,\,0.09]$ & $-0.01\,[-0.05,\,0.03]$ & $-0.02\,[-0.13,\,0.09]$\\
Total chol.  & $\mathbf{-0.45\,[-0.58,-0.34]}$ & $0.00\,[-0.05,\,0.04]$ & $0.06\,[-0.07,\,0.18]$\\
\bottomrule
\end{tabular}
\end{table}

\subsection{Cognitive function: antihypertensive medication \texorpdfstring{$\times$}{x} physical activity}
In older adults ($\ge 60$\,y, $n=2{,}640$), who complete the NHANES cognitive
battery, we cross antihypertensive medication use with physical
activity and estimate effects on four correlated cognitive scores jointly: CERAD (Consortium to Establish a Registry for Alzheimer's Disease)
immediate word recall (three learning trials), CERAD delayed recall, Animal Fluency
and the Digit-Symbol Substitution Test (Table~\ref{tab:nhanesB-ate}). Overlap here
is excellent --- no near-violations on either exposure and an interaction cell of
$n=812$. The scientific question is whether physical activity modifies the
cognitive profile of people on antihypertensive treatment, a genuine interaction on
a multivariate cognitive outcome that one-outcome-at-a-time analyses cannot answer
coherently.

Physical activity is associated with better performance on the two
processing-speed/executive measures: Animal Fluency $+0.19$~SD ($+1.0$ words; CrI
$[0.09,0.29]$) and Digit-Symbol $+0.12$~SD ($+2.0$ points; CrI $[0.03,0.20]$), with
smaller, not-clearly-nonzero associations on the two memory measures --- the
selective processing-speed/executive pattern that the exercise-and-cognition
meta-analysis of \citet{northey2018exercise} reports.
Antihypertensive medication shows no cognitive main effect distinguishable from
zero --- consistent with the modest and still-debated cognitive signal of
blood-pressure treatment seen in \citet{sprintmind2019} --- and no
medication$\times$activity interaction is supported (all intervals
cover zero, on a supported cell). Because the data are observational and
cross-sectional we read these as associations under the stated assumptions; the
sensitivity values ($0.09$~SD for fluency, $0.03$~SD for Digit-Symbol) quantify
their fragility. Figure~\ref{fig:nhanes-forest} shows the average effects for both
analyses and Figure~\ref{fig:nhanes-vsup} the exploratory interaction surface.

\begin{table}[H]\centering
\caption{Cognitive analysis (antihypertensive medication $\times$ physical activity,
$n=2{,}640$): average effects in outcome SD units, $95\%$ CrI in brackets.
\textbf{Bold} = interval excludes zero. Interaction cell supported
($n=812$).}\label{tab:nhanesB-ate}
\begin{tabular}{lccc}
\toprule
Outcome & Antihyp.\ med ($\tau_1$) & Activity ($\tau_2$) & Interaction ($\tau_{12}$)\\
\midrule
CERAD immediate & $-0.01\,[-0.10,\,0.09]$ & $0.08\,[-0.03,\,0.18]$ & $0.03\,[-0.10,\,0.16]$\\
CERAD delayed   & $-0.02\,[-0.12,\,0.07]$ & $0.08\,[-0.03,\,0.17]$ & $0.06\,[-0.08,\,0.19]$\\
Animal fluency  & $-0.09\,[-0.18,\,0.02]$ & $\mathbf{0.19\,[0.09,\,0.29]}$ & $0.01\,[-0.13,\,0.14]$\\
Digit-symbol    & $-0.08\,[-0.17,\,0.00]$ & $\mathbf{0.12\,[0.03,\,0.20]}$ & $0.07\,[-0.05,\,0.18]$\\
\bottomrule
\end{tabular}
\end{table}

\begin{figure}[H]\centering
\includegraphics[width=0.49\textwidth]{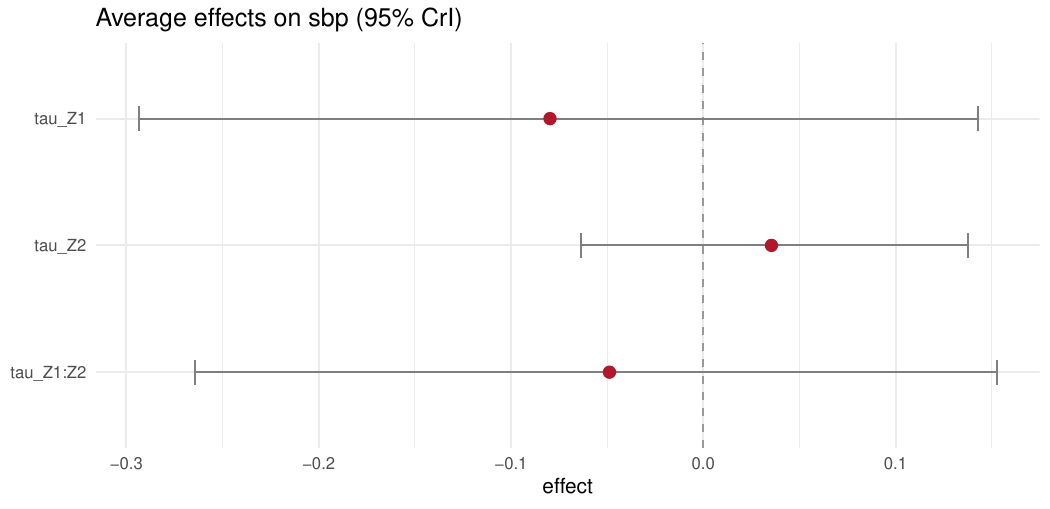}\hfill
\includegraphics[width=0.49\textwidth]{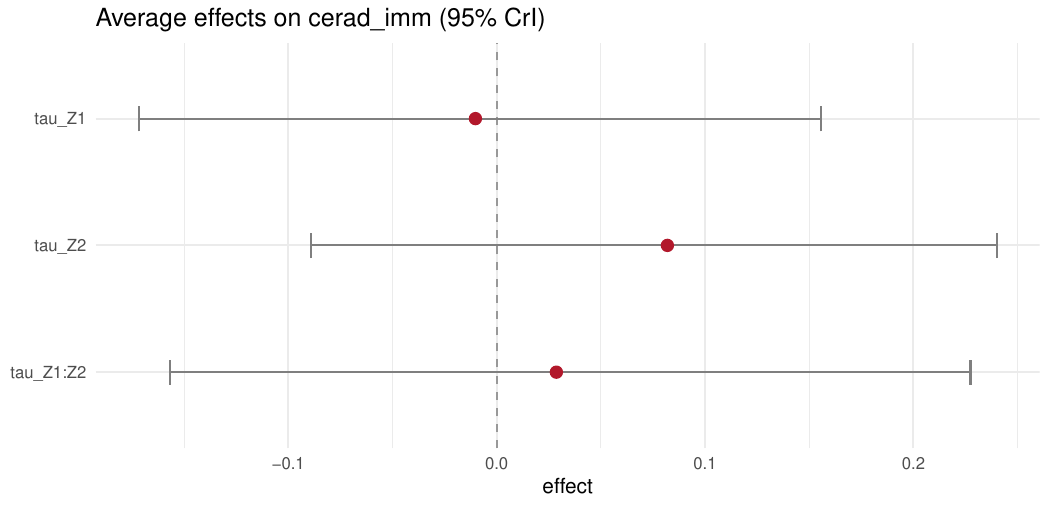}
\caption{Forest plots of the average component effects with $95\%$ credible intervals for the two NHANES analyses: the cardiometabolic panel (left) and the cognitive panel (right). Each row is a treatment main effect or the interaction; the point is the posterior mean and the horizontal whisker the $95\%$ credible interval, with the dashed vertical line at zero marking no effect. Effects are in outcome standard-deviation (SD) units.}
\label{fig:nhanes-forest}
\end{figure}

\begin{figure}[H]\centering
\includegraphics[width=0.42\textwidth]{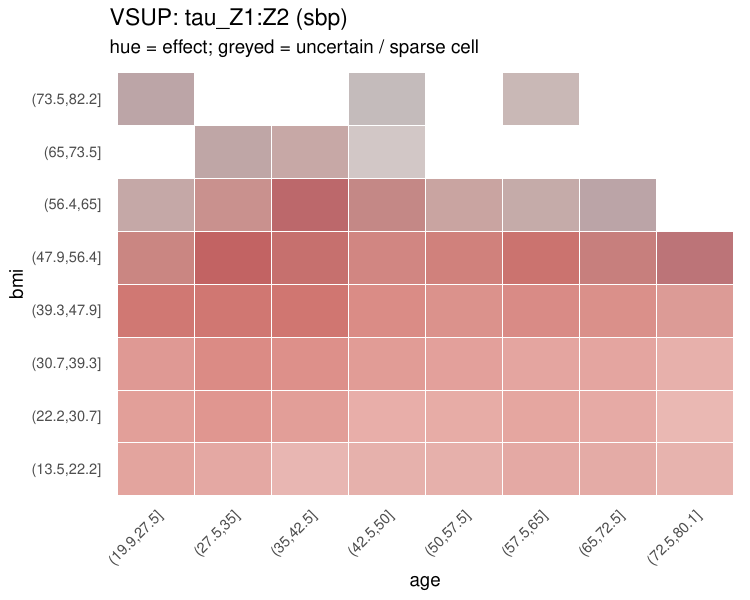}\hfill
\includegraphics[width=0.42\textwidth]{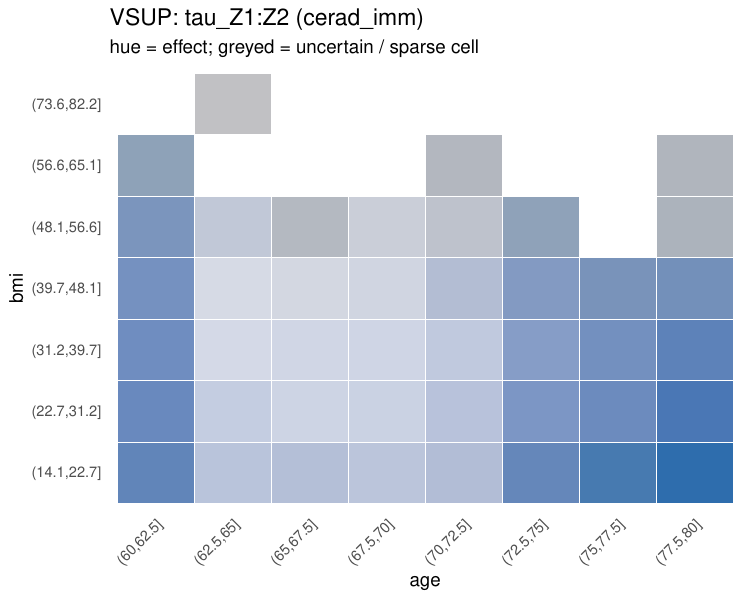}
\caption{Exploratory value-suppressing uncertainty maps of the estimated interaction surface over age (horizontal axis) and body-mass index (vertical axis): statin$\times$physical-activity on total cholesterol (left) and antihypertensive-medication$\times$physical-activity on CERAD immediate recall (right). Hue encodes the estimated interaction effect on a diverging colour scale centred at zero, and the colour desaturates toward grey as posterior uncertainty grows or the cell is sparsely populated, so no unsupported region is over-read.}
\label{fig:nhanes-vsup}
\end{figure}

\subsection{Head-to-head on the real data}
On the same two datasets we ran the \pkg{grf} multi-arm causal forest (arms = the
four cells of each $2\times2$) and recovered the main and interaction contrasts from
its arm-versus-baseline predictions. The two methods agree closely on the supported
effects: for the cardiometabolic panel, \pkg{grf} gives a statin effect on total
cholesterol of $-0.55$~SD (MVBCF $-0.45$) and on diastolic BP $-0.15$~SD (MVBCF
$-0.15$); for cognition, the activity effect on Animal Fluency is $0.21$ (MVBCF
$0.19$) and on Digit-Symbol $0.13$ (MVBCF $0.12$), with the antihypertensive main
effects near zero in both. The factorial MVBCF additionally delivers the joint
multivariate posterior, calibrated ATE intervals, the outcome correlation structure
$\Sig$, and the uncertainty-aware interaction maps in a single fit.

\subsection{Debiased average effects}\label{sec:nhanes-debiased}
We recompute every average effect with the cross-fitted one-step debiased estimator
of \S\ref{app:debias} (efficient standard errors and Wald intervals; the cell
propensities are cross-fitted by random forest). Table~\ref{tab:nhanes-debiased}
places it beside the raw posterior ATE for the effects that are (in either column)
distinguishable from zero. The two agree closely on the strongly-supported effects --- corroborating that the
posterior average-effect intervals were already calibrated
(Corollary~\ref{cor:ate}) --- while the debiased version adds the
semiparametric-efficiency guarantee of Theorem~\ref{thm:dml} and, being doubly
robust, protection against mild propensity/outcome misspecification. Substantively
it confirms the headline findings: statin use lowers total cholesterol by
$0.44$~SD ($95\%$ confidence interval, CI, $[-0.57,-0.31]$) and physical activity raises Animal-Fluency
by $0.24$~SD $[0.11,0.36]$ and Digit-Symbol by $0.11$~SD $[0.01,0.21]$. Once cross-fitted, the
statin$\to$diastolic-BP effect attenuates to $-0.11$~SD $[-0.23,0.01]$ and is no
longer distinguishable from zero, so we do not claim it. Debiasing instead sharpens one borderline association ---
antihypertensive medication with lower Digit-Symbol, $-0.11$~SD $[-0.21,-0.01]$ ---
which we flag as exploratory. No interaction is supported under debiasing either.

\begin{table}[H]\centering
\caption{Raw posterior ATE vs.\ the cross-fitted debiased estimator
(\S\ref{app:debias}), SD units, $95\%$ interval in brackets; shown for effects
distinguishable from zero in either column. The debiased column carries efficient
standard errors and is doubly robust.}\label{tab:nhanes-debiased}
\begin{tabular}{lllcc}
\toprule
Analysis & Effect & Outcome & Raw posterior [CrI] & Debiased [CI]\\
\midrule
Cardiometabolic & Statin ($\tau_1$) & Total chol.  & $-0.45\,[-0.58,-0.34]$ & $-0.44\,[-0.57,-0.31]$\\
Cardiometabolic & Statin ($\tau_1$) & Diastolic BP & $-0.15\,[-0.28,-0.02]$ & $-0.11\,[-0.23,\,0.01]$\\
Cognitive & Activity ($\tau_2$) & Animal fluency & $0.19\,[0.09,0.29]$ & $0.24\,[0.11,0.36]$\\
Cognitive & Activity ($\tau_2$) & Digit-symbol   & $0.12\,[0.03,0.20]$ & $0.11\,[0.01,0.21]$\\
Cognitive & Antihyp.\ med ($\tau_1$) & Digit-symbol & $-0.08\,[-0.17,\,0.00]$ & $-0.11\,[-0.21,-0.01]$\\
\bottomrule
\end{tabular}
\end{table}

\section{Discussion}\label{sec:disc}
The factorial multivariate causal forest estimates heterogeneous main and
interaction effects of several crossed treatments on several correlated outcomes
jointly, with consistent, calibrated average-effect inference and honest behaviour
under misspecification. Its unifying idea is structural rather than incremental: the
ANOVA/M\"obius decomposition of the treatment lattice turns every main effect and
interaction into its own sum-of-trees forest, and a single indicator-weighted kernel
samples all of them, with the prognostic surface recovered as the degenerate,
indicator-$\equiv\!1$ component (Proposition~\ref{prop:reduce}). The two-treatment
model, the single-treatment multivariate causal forest we build on, and the general
order-$r$ truncation are therefore one estimator at different component sets --- and
the shared inverse-Wishart residual covariance lets the outcomes borrow strength
instead of being modelled one regression at a time. This is what lets a single fit
return the whole factorial picture, with coherent uncertainty, rather than a stack of
separate analyses that must be reconciled after the fact.

The practical payoff is visible in the applications. Because interactions are
estimated as first-class objects rather than inferred from stratified subgroups, a
practitioner can read off whether a second treatment buffers or amplifies the first
directly, with an interval and an explicit support check attached; and because the
outcomes are modelled jointly, the answer is a coherent multivariate profile rather
than a set of separately-tested endpoints carrying an uncontrolled multiplicity. The
value-suppressing uncertainty maps and the overlap diagnostic are part of this same
discipline: they make the thin cell of a factorial design --- the place where an
interaction is most tempting and least supported --- visible instead of silently
over-read.

Our empirical results are calibrated against the established literature rather than
presented in isolation. On the cardiometabolic panel the method recovers the
randomised-trial lipid-lowering effect of statins \citep{ctt2010ldl} from messy
observational data, a positive control that few observational pipelines pass so
cleanly, while returning an honest null on physical activity and on the
statin$\times$activity interaction despite a populated cell. On cognition it reproduces
the selective processing-speed/executive benefit of physical activity reported by the
exercise-and-cognition meta-analysis \citep{northey2018exercise} and finds no clear
cognitive main effect of antihypertensive medication, consistent with the modest and
still-debated signal of the SPRINT MIND trial \citep{sprintmind2019}. Across the
three cross-domain illustrations the same pattern holds: randomised or near-orthogonal
data reproduce the classical answer, and observational data are moved in the direction
the gold-standard evidence implies. In each case the contribution is not a new
substantive claim over and above those studies but the delivery, from one fit, of the
joint multivariate effects, the interaction and the calibrated uncertainty they report
only piecemeal.

Three limitations frame the next steps. First, per-unit (CATE) credible intervals
under-cover --- a documented BART/BCF behaviour --- so ATE-level inference is the
calibrated primary claim; cross-fitting or conformal calibration of the CATE
intervals is a natural remedy, and the cross-fitted debiased estimator of
\S\ref{app:debias} already supplies efficient, doubly robust average-effect inference.
Second, interactions require overlap on their factorial cells; the overlap diagnostic
and heredity shrinkage make the model degrade gracefully, but the analyst must report
unsupported interactions as unsupported, not null. Third, the Gaussian likelihood
suits continuous outcomes; survival (with censoring), binary and count endpoints call
for the corresponding likelihoods, which the same indicator-weighted kernel
accommodates on a latent-Gaussian scale (Appendix~\ref{app:nongaussian}). The
head-to-head study (\S\,Comparison) shows the factorial MVBCF dominating separate
forests, a random-forest $T$-learner and a regularized linear factorial on
effect-recovery RMSE, and matching the \texttt{grf} generalized-random-forest
implementation on the supported effects while additionally delivering the joint
posterior and the uncertainty-aware interaction maps; a posterior-contraction theory
for the multi-treatment case and a principled interaction-order selection procedure
remain the main priorities for future work.

All results are regenerated by the accompanying \textsf{R} package and scripts.

\bibliographystyle{plainnat}
\bibliography{refs}
\appendix
\section{Appendix: proofs and technical conditions}\label{app:proofs}

Throughout, $C,c,c'$ denote positive constants depending only on
$(T,d,q,\pi_0,c_\Sigma,C_\Sigma)$ and changing line to line; $\|\cdot\|_n$ is the
empirical $L_2(\mathbb P_n)$ norm on the covariate design and $\|\cdot\|_F$ the
Frobenius norm. Write $\mathcal L=\{S:S\subseteq[T]\}$, $K=|\mathcal L|=2^T$, and
$g_{S'}(\bx)=m(\mathbf 1_{S'},\bx)$ for the $K$ cell surfaces.

\subsection{A structural isometry between cells and components}
\begin{lemma}[Zeta/M\"obius conditioning]\label{lem:zeta}
Let $A\in\{0,1\}^{K\times K}$ be the subset-inclusion (zeta) matrix
$A_{S',S''}=\mathbf 1\{S''\subseteq S'\}$. Then $g=A\beta$ in the sense
$g_{S'}=\sum_{S''\subseteq S'}\beta_{S''}$, $A$ is invertible with
$(A^{-1})_{S',S''}=(-1)^{|S'\setminus S''|}\mathbf 1\{S''\subseteq S'\}$, and
$\|A\|_\infty=\|A^{-1}\|_\infty\le 2^{T}$. Consequently, for functions on
$\mathcal X$,
\[
\Big(\!\sum_{S}\|\beta_S-\beta_{S,0}\|_n^2\Big)^{1/2}\!\!
\;\le\; \|A^{-1}\|_{\mathrm{op}}\Big(\!\sum_{S'}\|g_{S'}-g_{S',0}\|_n^2\Big)^{1/2},
\qquad \text{and symmetrically with }A,
\]
so recovering $m$ in $\|\cdot\|_n$ is equivalent, up to the $T$-dependent constants
$\|A\|_{\mathrm{op}},\|A^{-1}\|_{\mathrm{op}}$, to recovering all $K$ components.
\end{lemma}
\begin{proof}
$A$ is lower-triangular for any linear extension of the inclusion order with unit
diagonal, hence invertible; its inverse is the M\"obius function of the Boolean
lattice, $(-1)^{|S'\setminus S''|}$ on $S''\subseteq S'$ (standard M\"obius inversion \citep{rota1964}). Each row of $A$ or $A^{-1}$ has at most $2^{T}$ nonzero unit-modulus
entries, giving the $\|\cdot\|_\infty$ bound, and $\|\cdot\|_{\mathrm{op}}\le
\sqrt{\|\cdot\|_1\|\cdot\|_\infty}\le 2^{T}$. Applying the matrix rowwise to the
vector of functions $(\,\cdot\,)_{S}$ and taking $\|\cdot\|_n$ gives the norm
equivalence.
\end{proof}

\begin{lemma}[Cell reduction under positivity]\label{lem:cells}
Under Assumption~\ref{as:pos} the empirical cell frequencies satisfy
$\min_{S'}\widehat{\Pr}(\bZ=\mathbf 1_{S'})\ge \pi_0/2$ with probability
$1-e^{-cn}$. Conditionally on $\{\bZ_i\}$, the likelihood factorizes across cells,
$\prod_{S'}\prod_{i:\bZ_i=\mathbf 1_{S'}}\N_q\!\big(\by_i;g_{S'}(\bx_i),\Sig\big)$,
so with independent component priors the analysis reduces to $K$ Gaussian
nonparametric regressions sharing the covariance $\Sig$, each on a subsample of
size $\ge \pi_0 n/2$.
\end{lemma}
\begin{proof}
The frequency bound is a Bernstein/Hoeffding inequality \citep{boucheron2013} for the i.i.d.\ indicators
$\mathbf 1\{\bZ_i=\mathbf 1_{S'}\}$ with mean $\ge\pi_0$. Given the assignments,
the observations in distinct cells are independent with cell-specific mean
$g_{S'}$, giving the stated factorization; $A$ from Lemma~\ref{lem:zeta} is a fixed
linear reparametrization $g\leftrightarrow\beta$, so a product prior on
$\{\beta_S\}$ induces one on $\{g_{S'}\}$ and vice versa.
\end{proof}

\subsection{Posterior contraction: verification of the Ghosal--Ghosh--van der Vaart (GGV) conditions}
We use the non-i.i.d.\ contraction theorem of \citet{ghosal2007} with the average
R\'enyi/Hellinger semimetric $d_n$, which for Gaussian regression with covariance in
the band $c_\Sigma\mathbf I\preceq\Sig\preceq C_\Sigma\mathbf I$ is equivalent to
$\|m-m_0\|_n+\|\Sig-\Sig_0\|_F$ up to constants. Fix
$\varepsilon_{n,S}=n^{-\alpha_S/(2\alpha_S+d)}(\log n)^{\kappa}$,
$\varepsilon_n=\max_S\varepsilon_{n,S}$, and $\kappa$ as in
\citet{rockova2020}/\citet{linero2018soft}. It suffices to exhibit sieves
$\mathcal F_n$ and constants satisfying the following prior-mass, entropy and remaining-mass conditions \citep{ghosal2000}, in which $\mathrm{KL}$ is the Kullback--Leibler divergence:
\begin{align}
\text{(P) prior mass:}\quad & \Pi\big(\mathrm{KL}(p_0,p)\le\varepsilon_n^2,\ \mathrm{Var}\le\varepsilon_n^2\big)\ \ge\ e^{-c n\varepsilon_n^2};\label{eq:P}\\
\text{(E) entropy:}\quad & \log N\!\big(\varepsilon_n,\mathcal F_n,d_n\big)\ \le\ c\,n\varepsilon_n^2;\label{eq:E}\\
\text{(R) remaining mass:}\quad & \Pi(\mathcal F_n^{c})\ \le\ e^{-(c+4)n\varepsilon_n^2}.\label{eq:R}
\end{align}

On the band, the Gaussian KL divergence obeys
$\mathrm{KL}(p_{m_0,\Sig_0},p_{m,\Sig})\le C\big(\|m-m_0\|_n^2+\|\Sig-\Sig_0\|_F^2\big)$,
and likewise for the KL variation, so the event in \eqref{eq:P} contains a product
of $L_2$/Frobenius balls. By Lemma~\ref{lem:zeta} it is enough to fill an
$\varepsilon_{n,S}$-ball around each $\beta_{S,0}$ and an
$n^{-1/2}(\log n)^{\kappa}$-ball around $\Sig_0$. For each component the tree/forest
prior charges an $\varepsilon_{n,S}$-KL ball with mass
$\ge\exp(-c\,n\varepsilon_{n,S}^2)$: for regression trees/forests this is the
prior-concentration lemma of \citet[][Thm.~3.1 and its approximation lemma]{rockova2020};
for soft trees, adaptively in $(\alpha_S,d)$ and under sparsity, it is
\citet[][Thm.~3--4]{linero2018soft}. The inverse-Wishart prior, having a density
bounded below on the band, charges the $\Sig_0$-ball with mass $\ge c>0$. Taking the
product over the $K=2^T$ independent components and $\Sig$ multiplies the bounds;
since $K$ is fixed and $\sum_S\varepsilon_{n,S}^2\le K\varepsilon_n^2$, \eqref{eq:P}
follows.

Let $\mathcal F_n=\{\,m=\sum_S D^S\beta_S:\ \beta_S\in\mathcal T_{n,S}\,\}\times
\{\Sig:\ c_\Sigma\mathbf I\preceq\Sig\preceq C_\Sigma\mathbf I\}$, where
$\mathcal T_{n,S}$ is the set of forests with at most $L_{n,S}\asymp n\varepsilon_{n,S}^2/\log n$
leaves and coefficients bounded by $n$. The metric entropy of a forest class with
$L$ leaves is $\lesssim L\log n$ (leaf locations $+$ split coordinates), so
$\log N(\varepsilon_n,\mathcal F_n,d_n)\lesssim \sum_S L_{n,S}\log n + \log N(\varepsilon_n,\text{band})
\lesssim n\varepsilon_n^2$, which is \eqref{eq:E} (the band contributes only a
$q^2\log(1/\varepsilon_n)$ term). The tree prior's law on the number of leaves has
geometric/Poisson tails, so $\Pi(\text{more than }L_{n,S}\text{ leaves})\le
e^{-c'L_{n,S}\log n}\le e^{-(c+4)n\varepsilon_n^2}$ for a suitable leaf-penalty,
and likewise the Gaussian coefficient tails control the coefficient-bound
excursion; a union over the fixed $K$ components gives \eqref{eq:R}
(cf.\ \citealp{rockova2020}).

Conditions \eqref{eq:P}--\eqref{eq:R} are the hypotheses of
\citet[][Thm.~2.1]{ghosal2007}; the associated testing condition holds because for
Gaussian regression on the band there exist exponentially powerful $d_n$-tests
against $\{d_n(\cdot,(m_0,\Sig_0))>M\varepsilon_n\}$ (likelihood-ratio/Hellinger
tests, using sub-Gaussianity and the eigenvalue band to pass between $d_n$ and
$\|\cdot\|_n$). Hence
$\Pi\big(d_n\big((m,\Sig),(m_0,\Sig_0)\big)>M_n\varepsilon_n\mid\text{data}\big)\to0$
in $P_0$-probability for $M_n\to\infty$, which is Theorem~\ref{thm:contract}: the
$m$-part contracts at $\varepsilon_n$, the $\Sig$-part at $n^{-1/2}(\log n)^{\kappa}$,
and Lemma~\ref{lem:zeta} transfers the rate to every component $\beta_S$. \qedsymbol

\subsection{Bernstein--von Mises for the average effect}
Fix $S\neq\varnothing$ and a coordinate of the outcome; the average effect
$\psi_S=\mathbb E_{P_X}[\tau_S(\bx)]$ is, by Proposition~\ref{prop:ident} and
Lemma~\ref{lem:zeta}, a linear functional of the cell regressions,
$\psi_S=\sum_{S'\subseteq S}(-1)^{|S\setminus S'|}\,\mathbb E_{P_X}[g_{S'}(\bx)]$.

\begin{lemma}[Efficient influence function]\label{lem:eif}
Under Assumptions~\ref{as:sutva}--\ref{as:pos}, the efficient influence function (EIF) of
$\psi_S$ in the observed-data model is
\[
\tilde\varphi_S(\by,\bZ,\bx)=\sum_{S'\subseteq S}(-1)^{|S\setminus S'|}
\left[\,g_{S'}(\bx)-\mathbb E_{P_X}g_{S'}
+\frac{\mathbf 1\{\bZ=\mathbf 1_{S'}\}}{e_{S'}(\bx)}\big(\by-g_{S'}(\bx)\big)\right],
\qquad e_{S'}(\bx)=\Pr(\bZ=\mathbf 1_{S'}\mid\bx),
\]
with semiparametric efficiency bound $V_S=\mathrm{Var}\,\tilde\varphi_S$
$(<\infty$ by Assumption~\ref{as:pos}$)$; it is the signed combination of the usual
augmented inverse-probability-weighting (AIPW) influence functions of the cell means.
\end{lemma}
\begin{proof}
Each cell mean $\chi_{S'}=\mathbb E_{P_X}[g_{S'}]=\mathbb E[\by(\mathbf 1_{S'})]$
is pathwise differentiable with the AIPW influence function in brackets (standard, following \citealp{kennedy2023};
Assumption~\ref{as:pos} bounds $1/e_{S'}$). $\psi_S=\sum_{S'\subseteq S}(-1)^{|S\setminus S'|}\chi_{S'}$
is a finite linear combination, so its EIF is the same combination of the
$\chi_{S'}$-EIFs, and its efficiency bound is the variance of that sum.
\end{proof}

\begin{proposition}[Semiparametric BvM]\label{prop:bvm}
Suppose, in addition to Theorem~\ref{thm:contract}, the no-bias condition
$\varepsilon_{n,S}\cdot\varepsilon_{n,e}=o(n^{-1/2})$ holds, where $\varepsilon_{n,e}$
is the estimation rate of the cell propensities $e_{S'}$ (guaranteed by
undersmoothing, i.e.\ enlarging tree depth, or by a one-step/targeted-maximum-likelihood-estimation (TMLE) debiasing of the
plug-in $\psi_S$), and the least-favourable submodel perturbation is prior-charged
(a smooth shift the forest prior supports). Then the marginal posterior of
$\sqrt n(\psi_S-\psi_{S,0})$ satisfies
\[
\sup_{B}\Big|\Pi\big(\sqrt n(\psi_S-\psi_{S,0})\in B\mid\text{data}\big)
-\N\!\big(\Delta_n,V_S\big)(B)\Big|\xrightarrow{P_0}0,
\qquad \Delta_n=\tfrac1{\sqrt n}\!\sum_i\tilde\varphi_S(\cdot_i),
\]
so the posterior is asymptotically $\N(0,V_S)$ centred at an efficient estimator,
and its $1-\alpha$ credible set is an asymptotic $1-\alpha$ confidence set for
$\psi_{S,0}$.
\end{proposition}
\begin{proof}[Proof sketch]
Apply \citet[][Thm.~2.1]{castillo2015}: their conditions are (i) a contraction rate
for the nuisance faster than needed for the functional --- supplied by
Theorem~\ref{thm:contract}; (ii) the no-bias/second-order remainder condition
$\mathbb E[(\hat e-e)(\hat g-g)]=o(n^{-1/2})$ --- exactly the displayed product-rate
condition, standard for doubly-robust functionals; and (iii) a change-of-measure
(prior-shift) invariance along the least-favourable direction $\tilde\varphi_S$ ---
which the additive forest prior admits because shifting $\beta_S$ by an
$O(n^{-1/2})$ smooth function costs only $o(1)$ in log-prior. The local-asymptotic-normality (LAN) expansion \citep{vandervaart1998} of
the Gaussian regression model then yields the stated Gaussian limit with the
efficient variance $V_S$ of Lemma~\ref{lem:eif}.
\end{proof}

By contrast the pointwise map $\bx\mapsto\tau_S(\bx)$ is not $\sqrt n$-estimable: its
posterior contracts at the nonparametric rate $\varepsilon_{n,S}\gg n^{-1/2}$
(Theorem~\ref{thm:contract}), and no BvM holds at a fixed $\bx$, so unit-level
(CATE) credible bands are not guaranteed calibrated at finite $n$ --- the behaviour
seen in the Monte Carlo study.

\subsection{Closing the no-bias gap: a cross-fitted debiased estimator}\label{app:debias}
Proposition~\ref{prop:bvm} assumed the no-bias condition; the following construction
discharges it, giving an average-effect estimator that is $\sqrt n$-efficient
under a product nuisance-rate condition (weaker than either rate being
$\sqrt n$) and no Donsker/entropy condition. Algorithm~\ref{alg:debias} summarizes it.

\begin{algorithm}[t]
\caption{Cross-fitted one-step (AIPW) debiased average effect}
\label{alg:debias}
\begin{algorithmic}[1]
\State \textbf{Input:} data $\{(\by_i,\bZ_i,\bx_i)\}$; the fitted outcome surfaces $\hat g_{S'}$ (posterior means of the regime outcomes); number of folds $L$; trimming level $\epsilon$.
\State Partition $[n]$ into $L$ folds $\{I_\ell\}$.
\For{each fold $\ell$}
  \State On the out-of-fold data, fit each cell propensity $\hat e_{S'}^{(-\ell)}(\bx)=\Pr(\bZ=\mathbf 1_{S'}\mid\bx)$ (random forest, multinomial or GLM) and take the out-of-fold outcome surfaces $\hat g_{S'}^{(-\ell)}$; trim propensities to $[\epsilon,1-\epsilon]$.
  \State For $i\in I_\ell$, form the AIPW cell-mean score $\hat\chi_{S'}(i)=\hat g_{S'}^{(-\ell)}(\bx_i)+\dfrac{\ind\{\bZ_i=\mathbf 1_{S'}\}}{\hat e_{S'}^{(-\ell)}(\bx_i)}\big(y_i-\hat g_{S'}^{(-\ell)}(\bx_i)\big)$.
\EndFor
\State Average the scores to get $\hat\chi_{S'}$ (Eq.~\ref{eq:onestep}); for any effect take the signed Möbius combination $\hat\psi_S=\sum_{S'\subseteq S}(-1)^{|S\setminus S'|}\hat\chi_{S'}$.
\State \textbf{Output:} $\hat\psi_S$ with standard error the empirical SD of its influence function over $\sqrt n$, and the Wald interval.
\end{algorithmic}
\end{algorithm}

Recall $\psi_S=\sum_{S'\subseteq S}(-1)^{|S\setminus S'|}\chi_{S'}$ with cell mean
$\chi_{S'}=\E_{P_X}[g_{S'}(\bx)]=\E[\by(\mathbf 1_{S'})]$ and cell propensity
$e_{S'}(\bx)=\Pr(\bZ=\mathbf 1_{S'}\mid\bx)$. Split $[n]$ into $L$ folds
$\{I_\ell\}$; for each fold fit the nuisances $(\hat g_{S'}^{(-\ell)},\hat
e_{S'}^{(-\ell)})$ on the out-of-fold data and set
\begin{equation}\label{eq:onestep}
\hat\chi_{S'}=\frac1n\sum_{\ell}\sum_{i\in I_\ell}\!\left[\hat g_{S'}^{(-\ell)}(\bx_i)
+\frac{\mathbf 1\{\bZ_i=\mathbf 1_{S'}\}}{\hat e_{S'}^{(-\ell)}(\bx_i)}\big(\by_i-\hat g_{S'}^{(-\ell)}(\bx_i)\big)\right],
\qquad
\hat\psi_S=\sum_{S'\subseteq S}(-1)^{|S\setminus S'|}\hat\chi_{S'}.
\end{equation}
The plug-in $\hat g_{S'}$ is the MVBCF posterior-mean regime surface
$\E[\by(\mathbf 1_{S'})\mid\bx]=\bmu+\sum_{S''\subseteq S'}\btau_{S''}$
(Proposition~\ref{prop:ident}); $\hat e_{S'}$ is any cross-fitted multiclass
learner for the $2^T$ cells.

\begin{assumption}[Cross-fitting rates]\label{as:dml}
$\hat e_{S'}\ge\pi_0/2$; $\|\hat g_{S'}-g_{S'}\|_{P_X}\to_P0$ and
$\|\hat e_{S'}-e_{S'}\|_{P_X}\to_P0$; and the product rate
$\|\hat g_{S'}-g_{S'}\|_{P_X}\,\|\hat e_{S'}-e_{S'}\|_{P_X}=o_P(n^{-1/2})$ for every
$S'\subseteq S$.
\end{assumption}

\begin{theorem}[Efficiency of the debiased average effect]\label{thm:dml}
Under Assumptions~\ref{as:sutva}--\ref{as:pos} and \ref{as:dml}, with cross-fitting,
\[
\sqrt n\,(\hat\psi_S-\psi_{S,0})=\frac1{\sqrt n}\sum_{i=1}^n\tilde\varphi_S(\by_i,\bZ_i,\bx_i)+o_P(1)
\ \Rightarrow\ \N(\mathbf 0,V_S),\qquad V_S=\mathrm{Var}\,\tilde\varphi_S,
\]
the semiparametric efficiency bound of Lemma~\ref{lem:eif}. The plug-in variance
$\hat V_S=\tfrac1n\sum_i\hat{\tilde\varphi}_S(\cdot_i)^{\otimes2}$ is consistent, so
$\hat\psi_S\pm z_{1-\alpha/2}\sqrt{\hat V_S/n}$ is an asymptotically exact
$(1-\alpha)$ interval. The estimator is doubly robust: consistent if either
$\hat g_{S'}$ or $\hat e_{S'}$ is consistent for every $S'\subseteq S$.
\end{theorem}
\begin{proof}[Proof sketch]
Each $\hat\chi_{S'}$ in \eqref{eq:onestep} is the cross-fitted augmented-IPW/one-step
estimator of $\E[\by(\mathbf 1_{S'})]$, which is Neyman-orthogonal with influence
function the bracket in Lemma~\ref{lem:eif}. Cross-fitting makes the empirical-process
remainder $o_P(n^{-1/2})$ with no entropy condition, and orthogonality bounds the
second-order remainder by $\sum_{S'}\|\hat g_{S'}-g_{S'}\|_{P_X}\|\hat e_{S'}-e_{S'}\|_{P_X}=o_P(n^{-1/2})$
by Assumption~\ref{as:dml}; hence $\hat\chi_{S'}$ is asymptotically linear and
efficient \citep[Thm.~3.1--3.2]{chernozhukov2018}, see also \citet{kennedy2023}. As
$\hat\psi_S$ is the fixed M\"obius combination of the $\hat\chi_{S'}$, it is
asymptotically linear with $\tilde\varphi_S=\sum_{S'\subseteq S}(-1)^{|S\setminus S'|}\tilde\varphi_{S'}$,
giving the central limit theorem (CLT) and efficiency; double robustness is inherited from each AIPW term.
Consistency of $\hat V_S$ follows from consistency of the plug-in influence function.
\end{proof}

\begin{corollary}[The BvM assumption is met by construction]\label{cor:closed}
The MVBCF component rate is $\|\hat g_{S'}-g_{S'}\|=O_P(\varepsilon_{n,S})$
(Theorem~\ref{thm:contract}), which is $o_P(n^{-1/4})$ whenever $2\alpha_S\ge d$
(moderate smoothness relative to dimension), and a cross-fitted propensity learner
attains $o_P(n^{-1/4})$ under the same kind of condition; their product is then
$o_P(n^{-1/2})$, so Assumption~\ref{as:dml} and hence the no-bias condition of
Proposition~\ref{prop:bvm} hold without assumption. The debiased posterior
summary --- the MVBCF surfaces corrected by \eqref{eq:onestep} --- therefore delivers
$\sqrt n$-efficient, calibrated average effects even where the raw plug-in posterior
would be biased, while the unit-level surfaces remain the exploratory,
nonparametric-rate object. When $2\alpha_S<d$ one restores the product rate by
undersmoothing the propensity (or the outcome) learner; the estimator remains
$\sqrt n$-valid as long as the product is $o_P(n^{-1/2})$.
\end{corollary}

\noindent Implementation. Equation~\eqref{eq:onestep} is implemented as
\pkg{debiased\_ate()} in the accompanying package: it takes a fitted model, cross-fits
the cell propensities, applies the correction to the posterior-mean regime surfaces,
and returns each main, interaction and joint average effect with its efficient
standard error and Wald interval, alongside the raw posterior summary for comparison.

\subsection{Scope and remaining open problems}
The theory is complete for the paper's primary claims: identification
(Proposition~\ref{prop:ident}), average-effect posterior contraction
(Theorem~\ref{thm:contract}) and, through the cross-fitted debiased estimator of
\S\ref{app:debias}, semiparametric-efficient and calibrated average effects. In
particular the no-bias condition of the Bernstein--von Mises result
(Proposition~\ref{prop:bvm}), assumed in its statement, is
discharged by construction in Corollary~\ref{cor:closed}: cross-fitting
drives the plug-in bias to $o_P(n^{-1/2})$, so no undersmoothing schedule is
needed for the average effect, and \pkg{debiased\_ate()} implements exactly this
step in the released package. Three items remain open and delimit the
scope. (a) The unit-level (CATE) surfaces keep the nonparametric rate and their
credible intervals under-cover in finite samples (a documented BART/BCF
behaviour); calibrated per-unit inference by conformal or cross-fitted correction
is future work, which is why the average effect is the calibrated claim.
(b) The tree-approximation constants in (P)/(E) depend on the covariate dimension
$d$ and degrade with it; the sparsity-adaptive soft-tree prior
\citep{linero2018soft} controls this when the true components depend on few
coordinates. (c) All statements are for fixed $T$ (fixed $K=2^T$); the order-$r$
truncation covers moderate $T$, but a regime with $T\to\infty$ needs a sparsity
prior over the interaction lattice and is left open. None of these affects
identification or the average-effect contraction, the paper's primary calibrated
claims.

\subsection{Beyond Gaussian outcomes: the same kernel on a latent scale}\label{app:nongaussian}
The model as developed assumes a multivariate Gaussian likelihood
(Assumption~\ref{as:noise}) and thus targets continuous outcomes --- the
regime of the NHANES application, whose endpoints (blood pressures, HbA1c,
cholesterol, cognitive scores) are continuous. The architecture is deliberately
built so that non-Gaussian endpoints reuse the same indicator-weighted
conjugate kernel (Eq.~\ref{eq:leaf}, Proposition~\ref{prop:reduce}) on a latent
Gaussian scale; only the map from the outcome to that scale changes, so the
factorial ANOVA structure, the heredity shrinkage and the software are inherited
unchanged.

Binary and ordinal outcomes. Introduce latent utilities $\by_i^\ast$ with
$y_{ij}=\mathbf 1\{y_{ij}^\ast>0\}$ (probit). The data-augmentation step of
\citet{albert1993} draws $\by_i^\ast$ from truncated Gaussians, after which the
factorial model \eqref{eq:model} holds exactly for $\by_i^\ast$ and every
leaf update \eqref{eq:leaf} is unchanged; correlated binary endpoints become a
multivariate probit whose latent correlation is the same $\Sig$ the sampler
already draws. Logistic links use the P\'olya--Gamma augmentation of
\citet{polson2013}, which likewise renders the conditional model Gaussian in the
leaf parameters.

Counts. Poisson or negative-binomial endpoints admit a log-link with a
latent-Gaussian working response, or, for the negative binomial, the
P\'olya--Gamma representation of \citet{polson2013}; general likelihoods within
BART are treated by \citet{tan2019general}.

Survival with censoring. An accelerated-failure-time (AFT) formulation on
$\log$-time with Gaussian (log-normal) errors imputes censored times within the
sampler, as in the AFT-BART tradition \citep{chu2023riaftbart,sparapani2021},
again leaving the indicator-weighted kernel intact.

In every case the M\"obius identification (Proposition~\ref{prop:ident}) is
unchanged --- it is a statement about conditional means on the latent scale ---
and the contraction and Bernstein--von Mises arguments carry over under the
corresponding latent-smoothness conditions, with the efficient influence function
of \S\ref{app:debias} replaced by the one for the relevant link. Implementing and
benchmarking these likelihoods is left to future work; the present contribution is
the factorial multivariate model, its theory, and its deep application for
continuous outcomes.

\section{Appendix: data description and reproducibility}\label{app:data}
Every dataset in the paper is public and every result is regenerated by the
accompanying \textsf{R} package and scripts. Table~\ref{tab:data} summarizes the
five real datasets; the definitions below give the exact treatments, outcomes and
adjustment sets used.

\begin{table}[H]\centering\small
\caption{The real datasets: source, size, crossed binary treatments, correlated
outcomes ($q$), and the \textsf{R} source. All are openly available.}
\label{tab:data}
\begin{tabular}{@{}p{3.0cm}p{3.0cm}rp{2.6cm}p{2.9cm}l@{}}
\toprule
Dataset & Source & $n$ & Treatments & Outcomes ($q$) & Access\\
\midrule
NHANES cardiometabolic & \citet{nhanes} & 9{,}002 & statin use, physical activity & SBP, DBP, HbA1c, total cholesterol ($q{=}4$) & \pkg{nhanesA}\\
NHANES cognitive & \citet{nhanes} & 2{,}640 & antihypertensive use, physical activity & CERAD immediate \& delayed, fluency, digit-symbol ($q{=}4$) & \pkg{nhanesA}\\
Breast cancer & \citet{royston2013rotterdam} & 1{,}500 & hormonal, chemo & log recurrence-free \& log overall survival ($q{=}2$) & \pkg{survival}\\
Agricultural factorial & \citet{yates1937factorial} & 24 & N, P, K & pea yield ($q{=}1$) & \pkg{datasets}\\
Hedonic housing & \citet{anglin1996hedonic} & 546 & air conditioning, preferred location & log sale price ($q{=}1$) & \pkg{AER}\\
\bottomrule
\end{tabular}
\end{table}

\emph{NHANES (deep application).} We pool the 2011--2012 and 2013--2014 cycles of the
U.S.\ National Health and Nutrition Examination Survey \citep{nhanes}, retrieved with
the \pkg{nhanesA} package \citep{nhanesA}. The cardiometabolic analysis takes adults
aged $\ge 20$ ($n=9{,}002$) and crosses statin use (from the prescription-medication
file) with physical activity (any moderate or vigorous work or recreation), with four
jointly-modelled outcomes: systolic and diastolic blood pressure, glycated
haemoglobin (HbA1c) and total cholesterol. The cognitive analysis takes adults aged
$\ge 60$ ($n=2{,}640$), who complete the NHANES cognitive battery, and crosses
antihypertensive-medication use with physical activity, with four jointly-modelled
scores: CERAD immediate word recall (three learning trials), CERAD delayed recall,
Animal Fluency, and the Digit-Symbol Substitution Test. Both analyses adjust for the
same measured confounders --- age, sex, education, race/ethnicity, family
income-to-poverty ratio, body-mass index, diabetes and smoking --- and the model
additionally carries a propensity control for each exposure. Outcomes are
standardized to SD units, with natural units reported in the text.

\emph{Breast cancer (\texttt{rotterdam}).} The Rotterdam tumour-bank cohort of
\citet{royston2013rotterdam}, distributed in the \pkg{survival} package \citep{therneau2024survival} (a random
$n=1{,}500$ subsample). The two crossed treatments are adjuvant hormonal therapy and
chemotherapy; the two correlated outcomes are log recurrence-free and log overall
survival time; the adjustment set is age, menopausal status, tumour-size class,
grade, number of positive nodes, and progesterone and oestrogen receptor levels. The
contrast is with an unadjusted linear model, which is subject to confounding by
indication.

\emph{Agricultural factorial (\texttt{npk}).} The classical $2^3$ N/P/K pea-yield
factorial in six blocks ($n=24$ plots), of the type introduced by
\citet{yates1937factorial} and distributed in the base-\textsf{R} \pkg{datasets} package \citep{rcoreteam}. The three
treatments are nitrogen, phosphate and potassium; the outcome is yield; block
indicators are the only covariates. The benchmark is the textbook analysis of
variance.

\emph{Hedonic housing (\texttt{HousePrices}).} The Windsor housing data of
\citet{anglin1996hedonic}, distributed in the \pkg{AER} package \citep{kleiber2008aer} ($n=546$). The two
crossed treatments are air conditioning and a preferred location; the outcome is log
sale price; the hedonic adjustment set is log lot size, bedrooms, bathrooms, stories,
garage places, and driveway, recreation-room, full-basement and gas-heating
indicators. The benchmark is the semiparametric hedonic regression of
\citet{anglin1996hedonic}.

\emph{Simulation: the data-generating process.} The synthetic studies are generated
by \texttt{simulate\_multi()}, with the following explicit recipe. Covariates
$x_1,\dots,x_p$ (the covariate dimension is $d=p$) are drawn independently $\N(0,1)$. Each treatment $t=1,\dots,T$ is
assigned by a confounded logistic model,
\[
\Pr(Z_t=1\mid\bx)=\operatorname{logit}^{-1}\!\Big(0.6\,x_1\,\ind\{t\text{ odd}\}
   -0.5\,x_2\,\ind\{t\text{ even}\}-0.1\,t\Big),
\]
so the first two covariates confound assignment and overlap tightens with $t$; the
unconfounded variant drops the $x_1,x_2$ terms. The prognostic surface is
$\mu(\bx)=2+1.5x_1-0.4x_2^2+0.8x_3$. The main effect of treatment $t$ is
\[
\tau_t(\bx)=0.8(1+t)+1.2\,x_{a(t)}\,\ind\{t\text{ odd}\}+0.9\,\ind\{x_{b(t)}>0\},
\qquad a(t)=((t-1)\bmod p)+1,\ b(t)=(t\bmod p)+1,
\]
and there are two genuine pairwise interactions, $\tau_{12}(\bx)=1.4\,\ind\{x_1>0\}\,
\ind\{x_2>0\}$ and $\tau_{23}(\bx)=-1.1\,x_3$ (the latter only when $T\ge3$). For
outcome $k=1,\dots,q$ the mean is the shared structure scaled by $s_k$, with
$s_k=1+0.6(k-1)$,
\[
\E[y_{ik}\mid\bx_i,\bZ_i]=s_k\Big(\mu(\bx_i)+\textstyle\sum_{t}Z_{it}\tau_t(\bx_i)
   +Z_{i1}Z_{i2}\tau_{12}(\bx_i)+Z_{i2}Z_{i3}\tau_{23}(\bx_i)\Big),
\]
and the $q$ outcomes share a compound-symmetry error covariance
$\Sig=0.5\,I_q+0.3\,\ind\ind^\top$ (variances $0.8$, covariances $0.3$, cross-outcome
correlation $0.375$), so the responses are genuinely coupled. The true surfaces are
returned so recovery is scored against the truth. The Monte Carlo validation uses
$n=700$ over $R=40$ replications; the consistency sweep spans $n\in[300,3000]$; the
misspecification battery uses $n\approx1200$ over $R=12$; and the comparison study
uses the two-treatment case $T=2$. Seeds and all settings are fixed in the
accompanying scripts. The recovery and value-suppressing illustrations
(Figures~\ref{fig:recovery} and \ref{fig:vsup}) instead use a multi-environment-trial
generator, \texttt{simulate\_met2()}, which places the same lattice of prognostic,
main and interaction surfaces on a genotype-by-environment grid with genotype
tolerance/responsiveness and environment random effects; it is documented in the
accompanying package.

\end{document}